\documentclass[12pt, twoside]{article}
\usepackage{amsmath}
\usepackage{amsthm}
\usepackage[onehalfspacing]{setspace}
\usepackage{verbatim}
\usepackage{lmodern}

\usepackage[left=1.3in,right=1.3in,top=1.25in,bottom=1.25in]{geometry}

\usepackage{amsfonts}

\usepackage{qtree}

\usepackage{tikz}

\usepackage{multirow,array}

\usepackage[pdfstartview={XYZ null null 1},colorlinks=true,citecolor=blue,linkcolor=blue,breaklinks=true]
 {hyperref}

\usepackage{latexsym}
\usepackage{float}
\usepackage{graphicx}
\usepackage{enumitem}

\usepackage[font=small,labelsep=none]{caption}

\usepackage{natbib}

\usepackage[title]{appendix}
\usepackage{graphicx}

\usepackage{CJK}

\newtheorem{remark}{Remark}[section]

\newtheorem{conjecture}{Conjecture}[section]
\newtheorem{corollary}{Corollary}[section]

\newtheorem{definition}{Definition}[section]

\newtheorem{openq}{Open Question}[section]

\newtheorem{lemma}{Lemma}[section]

\newtheorem*{proposition*}{Proposition}

\newtheorem{proposition}{Proposition}[section]

\newtheorem{fact}{Fact}[section]

\theoremstyle{definition}

\theoremstyle{definition}

\title{Incompleteness, Independence, and Negative Dominance}
\date{\today}
\author{Harvey Lederman\footnote{\url{harvey.lederman@utexas.edu}. Thanks to Juan Dubra,  \"Ozg\"ur Evren, Brian Hedden, Giacomo Molinari, Efe Ok, Larry Samuelson, Dean Spears, Teru Thomas, and S\'{e}verine Toussaert for helpful discussion, correspondence, and encouragement.}}
\begin{document}

\maketitle

\begin{abstract}
\noindent This paper introduces the axiom of Negative Dominance, stating that, if a lottery $f$ is strictly preferred to a lottery $g$, then some outcome in the support of $f$ is strictly preferred to some outcome in the support of $g$. It is shown that, if preferences are incomplete on a sufficiently rich domain, then this plausible axiom, which holds for complete preferences, is incompatible with an array of otherwise plausible axioms for choice under uncertainty. In particular, in this setting, Negative Dominance conflicts with the standard Independence axiom. A novel theory, which includes Negative Dominance, and rejects Independence, is developed and shown to be consistent.

\


\

 \end{abstract}

\vfill 

\pagebreak

\section{Introduction}

Orthodox decision theory assumes that agents' preferences are \emph{complete}: between any two outcomes, or lotteries over outcomes, the agent either prefers one to the other, or is exactly indifferent between them. But it has long been recognized that this axiom may be doubted. In a famous passage, \citet[p. 446]{aumann1962utility} puts the point forcefully:
\begin{quote} Of all the axioms of utility theory, the completeness axiom is perhaps the most questionable. Like others of the axioms, it is inaccurate as a description of real life; but unlike them, we find it hard to accept even from the normative viewpoint. Does ``rationality'' demand that an individual make definite preference comparisons between all possible lotteries (even on a limited set of basic alternatives)? 
\end{quote}

Aumann's judgments about normative requirements on preferences have had many defenders (see, e.g. \citet{bewley1986knightian} (\citet{bewley2002knightian}), and \citet{mandler2005incomplete}, cf. \citet[p. 19]{von1944theory}). There is a natural conception of preferences  on which it is obvious that preferences are sometimes incomplete.\footnote{For discussion of the contrast between psychological preference (which is naturally taken to be incomplete) and revealed preference (which is not) and how this bears on the completeness assumption, see \citet{mandler2001difficult,mandler2005incomplete}. Incompleteness is not just intuitively natural; it has also been used to explain behavior such as status quo maintenance (\citet{mandler2004status}).} But if preferences can be incomplete, there is a question of what decision theory, descriptively and normatively, should be adopted. 

The dominant strain of work on this question has examined the possibility of giving utility-like representations of theories which preserve other standard axioms, but drop completeness (\citet{richter1966revealed,peleg1970utility,bewley1986knightian,seidenfeld1995representation,shapleybaucells,delamo2002open,ok2002utility,dubra2004expected,nau2006shape,evren2008existence,evren2011multiutility,ok2012incomplete,galaabaatar2012expected,galaabaatar2013subjective,riella2015representation,gorno2017strict,hara2019coalitional,mccarthy2021expected,ok2023lipschitz,borie2023expected}). When this approach is extended to cases involving uncertainty, the focus has mostly been on vector-based or ``multi-utility'' representations, where preferences are represented by vectors or sets of utility functions which value uncertain options at their expected values. All of these representations imply the standard Independence Axiom, stating that a lottery $f$ is preferred to a lottery $g$ if and only if for any lottery $h$ and any $\alpha \in (0,1]$, a lottery which gives $\alpha$ probability of $f$ and $(1-\alpha)$ probability of $h$ is preferred to a lottery which gives $\alpha$ probability of $g$ and $(1-\alpha)$ probability of $h$.

In this paper, I present a series of impossibility results showing that if preferences are incomplete over a sufficiently rich domain then this Independence axiom is incompatible with a new but normatively and behaviorally plausible Negative Dominance axiom. I take these results, most notably the main impossibility result, proposition \ref{main}, to motivate exploration of a new decision theory for incomplete preferences, which upholds the possibility of incompleteness on a rich domain of outcomes alongside Negative Dominance, while rejecting Independence. The second main result of the paper, proposition \ref{consistency}, establishes the consistency of such a theory.

The core idea of the impossibility results can be introduced by a simple example. A consumer is choosing over outcomes and is sensitive to various features of these outcomes, here represented as different ``dimensions''. The different dimensions of the outcomes could be abstract features---for instance, a tradeoff between flavor and cost of various foods, a tradeoff between location and style of work in a job---or they could be a bundle of commodities for which the consumer does not have a settled exchange rate. Here we consider a stylized example, with two dimensions understood as the numbers of apples and oranges. The consumer prefers more applies and more oranges. Given their current assets in apples and oranges (which we write $(0,0)$), for instance, they strictly prefer $(1,1)$---the outcome of gaining one apple and one orange. But they do not have even weak preferences over certain tradeoffs between apples and oranges. They do not weakly prefer a trade which would gain four apples in exchange for two oranges, $(4,-2)$ to their present holdings, and they also do not weakly prefer their current holdings to this trade. Similarly, they do not prefer or disprefer to gain four oranges in exchange for losing two apples $(-2,4)$. Their preferences are incomplete over these trades.

Now suppose that the consumer can choose to enter a lottery in which, with probability $\frac{1}{2}$, they will gain four oranges at a cost of two apples ($(4,-2)$), and with probability $\frac{1}{2}$, they will gain four apples, at the expense of two oranges ($(-2,4)$). 

If the consumer is at least approximately risk-neutral and does not have diminishing marginal utility across such low values of apples and oranges, one might think that the consumer will strictly prefer the lottery to their current holdings. The lottery yields an expected value of one orange, and an expected value of one apple. The consumer strictly prefers this bundle ($(1,1)$) to their current holdings. So, one might think, they will prefer the lottery. This thought is closely related to the dominant tradition of representing incomplete preferences with sets of utility functions, where uncertain options are valued at their expected value.

But this line of thought conflicts with another extremely plausible idea. The consumer does not prefer either of the outcomes of this lottery to their current holdings. They do not prefer $(4,-2)$ or $(-2,4)$ to their current holdings $(0,0)$. As a consequence, they should not strictly prefer the lottery, either. Somewhat picturesquely: the consumer will not take the lottery home with them; they will only take home whichever of its outcomes arises, and they prefer neither of these to their current holdings. This intuitive, and behaviorally plausible line of reasoning is codified in the ``Negative Dominance'' axiom, which states that, if a consumer does not strictly prefer any of the outcomes in the support of a lottery $f$ to any of the outcomes in the support of a lottery $g$, then they do not strictly prefer $f$ to $g$.\footnote{As I discuss in detail in the next section, this is closely related to, though weaker than, the ``Vague Sure Thing'' principle of \citet{manzini2008representation}.} If preferences are complete, this principle reduces to a universally accepted dominance axiom, stating that, if every outcome of $g$ is weakly preferred to every outcome of $f$, then $g$ is weakly preferred to $f$. When preferences are incomplete, there is no longer an entailment from the standard dominance axiom to Negative Dominance, but Negative Dominance seems no less plausible. Negative Dominance has exactly the same form as standard dominance principles. It differs only in that, where those principles guarantee preference over the lotteries given certain preferences over the outcomes, the new axiom infers a \emph{lack} of preference for the lottery, from a \emph{lack} of preference for all of its outcomes (hence its name). But, as we have seen, if this Negative Dominance principle is accepted, then the decision-maker cannot have a strict preference for the lottery over $(4,-2)$ and $(-2,4)$, so the expectational reasoning presented above must be rejected.

There are many ways to respond to the conflict between these two arguments, which we will consider in more detail below. But here we will take it to motivate the question: is there a plausible general theory of incomplete preferences of the kind described above, which also respects Negative Dominance? The first contribution of the present work is to identify the Negative Dominance axiom and display a series of limitative results stemming from it, most notably that, in the presence of natural constraints on the structure of incomplete preferences, it conflicts with the Independence axiom. The second is to establish the consistency of a reasonably strong theory which includes Negative Dominance, and respects these natural constraints.

Section \ref{r2} explores limitative results arising from Negative Dominance in a setting where ``dimensions'' of preference are represented by $\mathbb R$, and the space of outcomes is $\mathbb R^n$. \ref{basic} presents the basic setup of the paper, \ref{negdom} formalizes the example just presented informally; section \ref{unidimensional} presents the main limitative result, which focuses attention on the rejection of Independence. Sections \ref{conversepareto}, \ref{continuity}, and \ref{kcomp} generalize the result in several ways, weakening several background assumptions and building the case for dropping Independence. Section \ref{newtheory} presents a theory including Negative Dominance, with only a weakened form of Independence; section \ref{setsofutilities} compares the new theory to the dominant tradition of using sets of utilities to represent incomplete preferences. Section \ref{literature} discusses related work. 

Appendix \ref{qualitative}, extends the impossibility results of section \ref{r2} to a more general setting where dimensions of preference are understood qualitatively. Section \ref{basic2} develops the conceptual framework of qualitative outcomes, and sections \ref{negdom2},  \ref{unidimensional2} and \ref{continuity2} present qualitative analogues of the results proven in sections \ref{negdom}, \ref{unidimensional} and \ref{continuity} respectively.

\section{Impossibility Results}\label{r2}

\subsection{General Setup}\label{basic}

Let $O$ be a non-empty set of outcomes and $\succeq$ a preorder (reflexive and transitive relation) on $O$. We use $x \succ y$ (``$x$ is strictly preferred to $y$'') to abbreviate $x \succeq y \land \neg y \succeq x$; $x \sim y$ (``$x$ is indifferent to $y$'') to abbreviate $x \succeq y \land y \succeq x$; $x \parallel  y$ (``$x$ is comparable to $y$'') to abbreviate $x \succeq y \lor y \succeq x$; and $x \bowtie  y$ (``$x$ is incomparable to $y$'') to abbreviate $x \not \parallel  y$. The set of (finite-support) lotteries on $O$, $\Delta(O) \subset \mathbb [0,1]^O$ is the set of functions $f$ such that $| \{ o \mid f(o) \neq 0\} |$ is finite and $\sum_{o\in O} f(o)=1$. $O$ is isomorphic to the subset of $\Delta(O)$ consisting of characteristic functions of elements of $O$, and from now on by $O$ we mean that subset, which we assume inherits relation of $\succeq$ from $O$ itself in the obvious way. Although technically all elements are drawn from $\Delta(O)$, I will often abuse notion and use the names of elements of $O$ (e.g., if $O= \mathbb R^2$, `$(1,1)$') as names for their characteristic functions. 

For $\alpha \in [0,1]$, and $f, g \in \Delta (O)$ we write write $f \alpha g$ to mean $\alpha f + (1-\alpha) g$, where the multiplication is understood pointwise, so that for all $o \in O$, $(f \alpha g)(o) = \alpha f(o) + (1-\alpha) g(o)$. 

The first goal of the paper is to examine whether a reasonable $\succeq$ can be extended from $O$ to a `reasonable' preorder (both reflexive and transitive) on $\Delta(O)$. 

\subsection{Negative Dominance and Expectationalism}\label{negdom}

I assume throughout the main text that $O=\mathbb R^2$; for simplicity I will write just $\Delta$ for $\Delta(\mathbb R^2)$.  Although I will work in $\mathbb R^2$, the limitative results generalize completely straightforwardly to $\mathbb R^n$.

The assumption that dimensions are well-modeled by $\mathbb R$ may be a plausible idealization in certain  contexts. 
 But assuming even this amount of structure in dimensions may not be plausible for all contexts, and in Appendix \ref{qualitative}, I prove analogous limitative results in a more abstract qualitative setting. 

For concreteness, I will temporarily assume the following about the ordering $\succeq$ on $\mathbb R^2$:
\begin{description}
\item[Pareto] If $x\geq x'$ and $y \geq y'$ then $(x,y) \succeq (x',y')$
\item[Converse Pareto] If $(x,y) \succeq (x', y')$ then $x\geq x'$ and $y \geq y'$.
\end{description}
Given Pareto, we have that $(0,0)\prec (3,3)$, and also $(0,0)\prec (0,4)$ and $(0,0)\prec (4,0)$. Given Converse Pareto we have that $(0,0) \bowtie  (4,-2)$ (recall that `$\bowtie $' stands for incomparability) and $(0,0) \bowtie  (-2,4)$. Converse Pareto is implausible for many applications; it amounts to the claim that tradeoffs in dimensions \emph{always} yield incomparability. I only assume it temporarily to allow a crisp statement of some initial results. Later, I provide a family of plausible theories which substantially relax the assumption, allowing significant tradeoffs, in section \ref{conversepareto}, where it is shown that the result still holds even in this more general setting. (See also below, remark \ref{conversepareto1} and the subsequent paragraph.)

An obvious first idea for developing an extension of this order from outcomes to lotteries is to impose axioms which ensure that for every $f \in \Delta$ there is a ``certainty-equivalent'', that is, an $o \in O$ such that $f \sim o$. By defining such certainty-equivalents, the space of lotteries is identified with a sub-space of the space of outcomes, and $\succeq$ can be automatically extended to $\Delta$. The most obvious way of assigning such certainty-equivalents is to take the coordinate-wise expectation of $f$. Using $\pi_1$, $\pi_2$ as projection operations on ordered pairs (so that $\pi_1 (x, y)=x$ and $\pi_2 (x, y) =y$) we identify each lottery with the outcome $(\sum_{o\in O} f(o) \pi_1(o), \sum_o f(o) \pi_2(o))$, hereafter denoted $exp(f)$. So, formally, the assumption is:
\begin{description}
\item[Expectationalism] For all $f \in \Delta$, $f \sim exp(f)$. 
\end{description}
Thus, for example, $(0,0)\frac{1}{2}(4,4)\sim(2,2)$. (Recall that the names for ordered pairs here stand for characteristic functions.) To simplify notation in giving examples, I sometimes denote uniform lotteries over a sequence of outcomes using $.../.../...$. If there are only two outcomes (as usual) the lottery assigns probability $\frac{1}{2}$ to each; if there are $3$, $\frac{1}{3}$, and so on. Thus we can write this claim as $(0,0)/(4,4) \sim (2,2)$, and as a result we have that $(0,0) \prec (2,2) \sim (0,0)/(4,4) \prec (3,3)$.

In saying that this axiom is natural, I assume that that the decision-maker does or might value the objective quantities in our dimensions linearly, and that the decision-maker is neither risk-prone nor risk averse. Later, I'll show that relaxing  these assumptions does not eliminate the basic results.

An outcome $o$ is \emph{in the support} of a lottery $f$ if $f(o) > 0$. Our next axiom is
\begin{description}
\item[Negative Dominance] If $f \succ g$, then for some $o, o'$ in the support of $f, g$ respectively, $o \succ o'$.
\end{description}
This axiom holds for every plausible decision theory for complete preferences, as can be seen by considering its equivalent contrapositive.
\begin{description}
\item[Negative Dominance (Contrapositive)] If for every $o, o'$ in the support of $f$ and $g$ respectively, $o\not \succ o'$, then $f \not \succ g$.
\end{description}
For complete preferences, lack of a strict preference implies the presence of a reverse weak preference, so we can restate this version of Negative Dominance as: ``if for every $o, o'$ in the support of $f, g$ respectively $o \preceq o'$, then $f \preceq g$.'' This is an extremely weak dominance principle, which says that if the worst outcome of $g$ is at least as good as the best outcome of $f$, then $g$ is better than $f$. This axiom clearly holds in every reasonable decision theory for complete preferences, and thus helps to motivate Negative Dominance. 

A different, but related, way of motivating Negative Dominance is by analogy the following very weak dominance principle:
\begin{description}
\item[Outcome Dominance] If for every $o,o'$ in the support of $f$ and $g$ respectively, $o\succ o'$, then $f\succ g$.
\end{description}
The contrapositive of Negative Dominance says that if there are \emph{no} $o, o'$ in the support of $f$ and $g$ such that $o \succ o'$, then also $f \not \succ g$. This analogy to Outcome Dominance is the reason for the name ``Negative Dominance''.

We now formalize the result in the introduction:

\begin{proposition}\label{expectationalism} Let $O=\mathbb R^2$. Pareto, Converse Pareto, Expectationalism and Negative Dominance are inconsistent. \end{proposition} 

To repeat: Converse Pareto is not essential to this result, and is just used here to simplify the statement. I will return to this point in more detail in a moment.

The proof is as above. The expectation of the lottery $f^*=(4,-2)/(-2,4)$ is $(1,1)$ and so Expectationalism gives us that $f^* \sim (1,1)$. For any preorder if $a \sim b$ and $b \succ c$, then $a \succ c$, so Pareto gives us that $f^*  \succ (0,0)$. Meanwhile by Converse Pareto $(4,-2) \not \succ (0,0)$ and $(-2,4) \not \succ (0,0)$ in violation of Negative Dominance.\footnote{This result, and in fact all of the limitative results in the paper featuring Negative Dominance, could be run using a substantial weakening of Negative Dominance:
\begin{description}
\item[Outcome Comparability] If $f \succ o^*$ or $o^* \succ f$, then there is some $o$ in the support of $f$ such that $o \parallel  o^*$.
\end{description}
This principle restricts Negative Dominance to preferences between lotteries and outcomes. It also replaces the requirement there be some outcome $o$ in the support of $f$ such that $o \succ o^*$, with the much weaker requirement that there merely be an $o$ which is \emph{comparable} to $o^*$, allowing (unlike Negative Dominance) that $o$ might be only weakly preferred to $o^*$. (Technically it also allows that the only such $o$ could be strictly dispreferred as well, but this is not of conceptual interest.)

In the case above, we have $(4,-2) \bowtie  (0,0)$, and also $(-2,4) \bowtie  (0,0)$, in violation of Outcome Comparability. Similar points hold for subsequent results as well, but I will not spell them out in detail.}

\begin{remark}\label{conversepareto1} Converse Pareto rules out non-trivial tradeoffs of any kind, which might seem to prevent any real-life economic applications of the setup in this result. More precisely, those who endorse the rationality of incomplete preferences generally hold that, for generic $o,o'$ with $o \bowtie  o'$ there is typically a \emph{sufficiently great} improvement in some dimension of $o'$ so that the resulting $o'^+ \succ o$. Converse Pareto rules out this possibility. But Converse Pareto is inessential to the result, as I discuss in much more detail in section \ref{conversepareto}. Suppose $\succ$ is defined so that, for some $o$ there are $a$ $b$ such that $o \bowtie  a$, $o \bowtie  b$ and the line between $a$ and $b$ intersects either $\prec_o=\{ o' \mid o \prec o'\}$ or $\succ_o=\{ o' \mid o \succ o'\}$ at some point $o^*$. (Equivalently, given that $o \sim o'$ iff $o=o'$, we may assume that $\bowtie_o=\{ o' \mid o \bowtie o' \}$ is non-convex.) Expectationalism allows us to construct a lottery $f$ supported on $a$ and $b$ such that $f \sim o^*$, contradicting Negative Dominance.
 \end{remark}

This Remark shows that Converse Pareto can be relaxed to allow non-trivial tradeoffs across dimensions, but it highlights that proposition \ref{expectationalism} does depend on assuming that for some $o$ the set $\bowtie_o$ is non-convex. So we can think of the result as demonstrating that three plausible constraints cannot be jointly satisfied: this non-convexity assumption; Negative Dominance; and Expectationalism. 

Below, however, I will mainly focus on the possibility of rejecting Negative Dominance or Expectationalism. Partly this is because, as I will discuss in a moment, \citet{manzini2008representation} have given us a good understanding of what happens if it is required require that the non-comparable region be convex, while we have a less clear understanding of how things look under our non-convexity assumption. But there is strong conceptual motivation for this approach, as well. The primary motivation for incomplete preferences in our multi-dimensional setting is the idea that the decision-maker does not have preferences over every tradeoff across different dimensions of outcomes. Provided some decision-maker can rationally treat distinct dimensions symmetrically (so that there preferences are symmetric about the line $x=y$), then non-convexity follows from non-trivial incompleteness (given that the agent's weak preferences on outcomes are topologically closed). Moreover, it is extremely plausible that the decision-maker's preferences can be incomplete over tradeoffs involving gains in the $x$ dimension and losses in the $y$ dimension if and only if they can also be incomplete on tradeoffs involving gains in the $y$ dimension and losses in the $x$ dimension. But if the region of noncomparabiilty is convex, this claim must fail.

We can further motivate the non-convexity assumption by considering an observation of \citet{manzini2008representation}. They study the following axiom, which is strictly stronger than Negative Dominance:
\begin{description}
\item[Vagueness Sure Thing ( VST)] If $f_1 \bowtie g_1$ and $f_2 \bowtie g_2$, then $\alpha f_1 + (1-\alpha) f_2 \bowtie \alpha g_1 + (1-\alpha) g_2$.
\end{description}
Manzini and Mariotti show that this axiom together with a package of other natural ones ensures the convexity of non-comparable lotteries (as opposed to outcomes). But as they note, there are settings where this axiom is unnatural. If there are outcomes $a, a^+, b, b^+$, so that $a^+ \succ a$ and $b^+ \succ b$, but $a, a^+ \bowtie b$ and $b, b^+ \bowtie a$, then although $a^+/b^+$ stochastically dominates the lottery $a/b$, VST will require that they be incomparable: VST is incompatible with the decision-maker having a strict preference for stochastically dominant lotteries. Our multidimensional outcome space in  makes it natural to hold that there are quadruples of outcomes like $a,a^+,b, b^+$ (e.g. $(-2,3),(-2,4),(3,-2),(4,-2)$), and so, make it natural to reject VST.

Negative Dominance is weaker than VST in two different ways. First, (and less important for our purposes here), the pairwise comparisons in the antecedent of VST apply to all lotteries, not just to outcomes. Second, and much more importantly, Negative Dominance yields non-comparability between lotteries $f$ and $g$ if \emph{every} $o, o'$ where $o$ is in the support of $f$ and $o'$ is in the support of $g$, are non-comparable. VST applies much more generally than this: it applies if there is just some way of pairing up outcomes so that \emph{each pair} is non-comparable. It is this latter generality which yields the conflict with stochastic dominance, given the structure of non-comparability just mentioned. In fact I will show below, in  proposition \ref{consistency} that Negative Dominance is compatible with stochastic dominance even given such a structure of preference. 

So, while I will discuss several ways of weakening Converse Pareto in section \ref{conversepareto}, from now on I will assume that the region of non-comparable outcomes can be non-convex, referring the reader who is interested in dropping this assumption to \citet{manzini2008representation}. In general I will continue to describe the above conflict as arising from only Negative Dominance and Expectationalism, no longer explicitly mentioning the non-convexity assumption in spite of its key role in the result.

Given this background, it is natural to blame the conflict in proposition \ref{expectationalism} on Expectationalism. The problem, intuitively, is that the expectation ``forgets'' how values are distributed across outcomes. The values in the different dimensions were ``mixed and matched'' in the relevant outcomes (a high $x$, with a low $y$, and vice-versa), but the expectation aggregates the outcomes in a way that loses this information. One can see this vividly by considering the fact that the expectation of our lottery $f^*=(4,-2)/(-2,4)$ is the same as that of a very different lottery $g^*=(4,4)/(-2,-2)$. Plausibly, the agent could strictly prefer $g^*$ to $(0,0)$, without strictly preferring $f^*$ to $(0,0)$. Part of our goal in the rest of the paper will be to develop a theory which allows this result.

\subsection{Unidimensional Expectations and Independence}\label{unidimensional}

A natural weakening of Expectationalism arises by restricting the class of lotteries which are valued expectationally to those which are ``unidimensional'', that is, lotteries where every outcome in a given lottery is the same on \emph{all but one dimension}. Graphically in $\mathbb R^2$, this would mean that if every outcome in a lottery lies on a single vertical (resp. horizontal) line, then the lottery is valued expectationally.

\begin{definition} An outcome $o$ is \emph{unidimensional with} an outcome $o'$ if and only if there is at most one $i \in \{1,2\}$ such that $\pi_i(o) \neq \pi_i (o')$. 

A lottery $f$ is \emph{unidimensional with} $g$ if every outcome in the support of $f$ is unidimensional with every outcome in the support of $g$.

A lottery is \emph{unidimensional} if every outcome in the support of the lottery is unidimensional with every other.  \end{definition}

Our assumption then says:
\begin{description}
\item[Unidimensional Expectations] If $f$ is unidimensional, then $f \sim exp(f)$.
\end{description}

Expectationalism is a natural assumption, but on reflection it is not obvious that a rational decision-maker may satisfy it. By contrast, insofar as we can speak of dimensions of outcomes, with preferences sensitive to them, it must be permitted for a rational decision-maker to satisfy Unidimensional Expectations. We can see this in two ways. 

First, if a decision-maker's preferences are linear in the objective dimensions of outcomes (in the sense that, holding one dimension fixed, they linearly value increases in the other), and if the decision-maker is not risk-averse or risk prone, then they must satisfy Unidimensional Expectations. Since both of these assumptions are clearly permitted for rational decision-makers (even if not mandatory), it is permitted for a rational agent to satisfy Unidimensional Expectations. (Note that, the assumptions that decision-maker's preferences are linear in the objective dimensions of outcomes, and that the decision-maker is not risk-averse or risk prone, do not together imply Expectationalism.) 

Second, the numerical values in each of our dimensions may be seen as deriving from a utility representation of the decision-maker's preferences over some array of objective features. Working for example in the setting of \citet{von1944theory}, these numerical utility representations for individual dimensions make sense only if we have a well-behaved ordering of unidimensional lotteries. On this approach, the meaningfulness of real-valued dimensions requires the satisfaction of Unidimensional Expectations.

So, unlike Expectationalism, which we saw reason to doubt above, Unidimensional Expectations seems obviously part of a possible description of a rational agent. But I will now show that an agent who satisfies Unidimensional Expectations must either violate Negative Dominance or violate Independence (assuming, as always, the non-convexity of the region of non-comparability).

Formally, the Independence axiom is as follows:
\begin{description}
\item[Independence] For any $h \in \Delta$ and $\alpha \in (0,1)$ $f \succeq g$ if and only if $f \alpha h \succeq g \alpha h$.
\end{description}
This principle was key to von Neumann-Morgenstern's representation result, and endorsed also by \citet{savage1972foundations}. But it was famously rejected by \citet{allais1953comportement} (cf. \citet{machina1982expected,buchak2013risk}), and its rejection has been central to the development of alternatives to expected utility reasoning in various contexts. I will suggest below that the ways in which Independence must fail in the present setting if Negative Dominance is upheld are not obviously connected to the ways in which it fails in theories of risk aversion (section \ref{newtheory}); a goal of the paper is to identify and explore these intuitive failures of Independence. 

The proof of the following result contains the core idea of other limitative results in the paper:

\begin{proposition} \label{main} Let $O= \mathbb R^2$. Pareto, Converse Pareto, Unidimensional Expectations, Independence, and Negative Dominance are inconsistent. \end{proposition}

\begin{proof} By Unidimensional Expectations, for any $a$, $$(-a,0)/(a,0) \sim (0,0)  \textrm{ \ and \  } (0,a)/(0,-a)\sim (0,0).$$ By Independence, we first derive that $(0,0) = (0,0)/(0,0) \sim [(-a,0)/(a,0)]/(0,0)$. A second application of Independence then gives us that that $$(0,0)/[(-a,0)/(a,0)] \sim f=[(0,a)/(0,-a)]/[(-a,0)/(a,0)].$$ By transitivity of $\sim$, we then have that $f \sim (0,0)$. (See the leftmost image in the figure for a graphical illustration of $f$.) We now assume that $a>0$. By Pareto $(-a,a)\succ (-a,0)$. So, by Independence a lottery $f^+$ which replaces the latter with the former in $f$, will be better than $f$ and hence better than $(0,0)$. That is, $$f^+=(0,a)/(a,0)/(-a,a)/(0,-a) \succ (0,0).$$ By another application of Independence, since $(a,-a) \succ (0,-a)$ a lottery $f^{++}$ which replaces the latter with the former in $f^+$ will again be better than $f^+$ and hence better than $f$ and hence better than $(0,0)$. That is, we have (as in the middle image below) $$f^{++} = (0,a)/(a,0)/(-a,a)/(a,-a)\succ (0,0).$$ 
By Unidimensional Expectations, we have that $(0,a)/(-a,a) \sim (\frac{-a}{2},a)$, and that $(a,0)/(a,-a)\sim(a,\frac{-a}{2})$. By two applications of Independence, we have that (as in the rightmost image below) $$f^{++}\sim(\frac{-a}{2},a)/(a,\frac{-a}{2}).$$
From this it follows that $(\frac{-a}{2},a)/(a,\frac{-a}{2})\succ(0,0)$. But, by Converse Pareto. $(\frac{-a}{2},a)\bowtie  (0,0)$ and also $(a,\frac{-a}{2}) \bowtie  (0,0)$, so this strict preference contradicts  Negative Dominance.\end{proof}

\noindent\makebox[\textwidth]
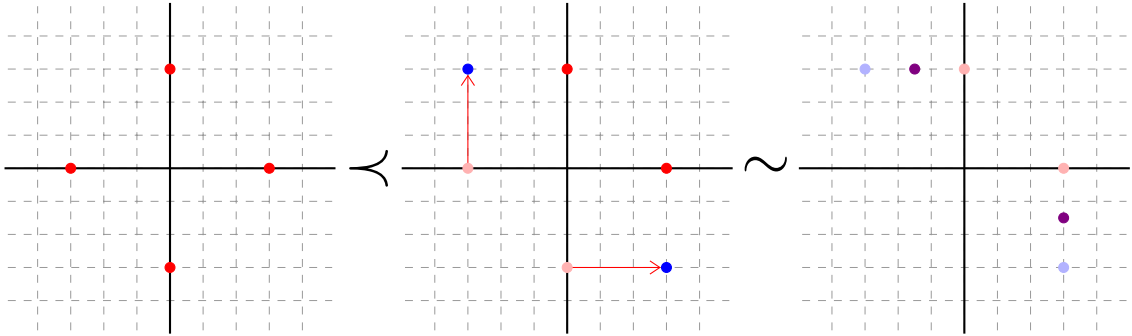
{
\begin{tikzpicture}[ampersand replacement=\&,scale=.44]
\draw[help lines, color=gray!80, dashed] (-4.9,-4.9) grid (4.9,4.9);
\draw[-, thick] (-5,0)--(5,0);
\draw[-, thick] (0,-5)--(0,5);
\filldraw [thin, red,fill=red] (3,0) circle[radius=1.5mm];
\filldraw [thin, red,fill=red] (0,3) circle[radius=1.5mm];
\filldraw [thin, red,fill=red] (-3,0) circle[radius=1.5mm];
\filldraw [thin, red,fill=red] (0,-3) circle[radius=1.5mm];
\draw (6,0) node[scale=2]{$\prec$};
\begin{scope}[shift={(12,0)}]
\draw[help lines, color=gray!80, dashed] (-4.9,-4.9) grid (4.9,4.9);
\draw[-, thick] (-5,0)--(5,0);
\draw[-, thick] (0,-5)--(0,5);
\draw[-, thin, red] (-3,0)--(-3,2.8);
\draw[-, thin, red] (-3.2,2.5)--(-3,2.8)--(-2.8,2.5);
\draw[-, thin, red] (0,-3)--(2.8,-3);
\draw[-, thin, red] (2.5,-2.8)--(2.8,-3)--(2.5,-3.2);

\filldraw [thin, red!30!white,fill=red!30!white] (-3,0) circle[radius=1.5mm];
\filldraw [thin, red!30!white,fill=red!30!white] (0,-3) circle[radius=1.5mm];
\filldraw [thin, red,fill=red] (3,0) circle[radius=1.5mm];
\filldraw [thin, red,fill=red] (0,3) circle[radius=1.5mm];
\filldraw [thin, blue,fill=blue] (-3,3) circle[radius=1.5mm];
\filldraw [thin, blue,fill=blue] (3,-3) circle[radius=1.5mm];

\draw (6,0) node[scale=2]{$\sim$};
\end{scope}
\begin{scope}[shift={(24,0)}]
\draw[help lines, color=gray!80, dashed] (-4.9,-4.9) grid (4.9,4.9);
\draw[-, thick] (-5,0)--(5,0);
\draw[-, thick] (0,-5)--(0,5);
\filldraw [thin, violet,fill=violet] (3,-1.5) circle[radius=1.5mm];
\filldraw [thin, violet,fill=violet] (-1.5,3) circle[radius=1.5mm];

\filldraw [thin, red!30!white,fill=red!30!white] (3,0) circle[radius=1.5mm];
\filldraw [thin, red!30!white,fill=red!30!white] (0,3) circle[radius=1.5mm];
\filldraw [thin, blue!30!white,fill=blue!30!white] (-3,3) circle[radius=1.5mm];
\filldraw [thin, blue!30!white,fill=blue!30!white] (3,-3) circle[radius=1.5mm];
\end{scope}
\end{tikzpicture}}

\begin{remark}\label{uneven} In the proof, we used the same value $a$, $-a$ for the $x$ and $y$-coordinates of our starting four points. But it is easy to see that we could have used $a \neq b$ for $x$ and $y$-coordinates, considering instead $(a,0),(-a,0), (0,b), (0,-b)$. This will be important later on. \end{remark}

Subsequent impossibility results will primarily be applications of proposition \ref{main}. But there is a second impossibility result which is also worth discussing in this context. 

We begin by introducing the key axiom. For $i \in \{1,2\}$, and any lottery $f$, let $\pi_i(f):\mathbb R \to [0,1]$ be the function  so that $\pi_i(f)(x)=\sum_{\{o \in O\mid \pi_i (o)=x\}} f(o)$. Given this definition, $\pi_i(f)$ can be thought of as the projection of a lottery onto the relevant axis. The proposed axiom then says that the projection of a lottery on an axis suffices to determine its certainty-equivalent in that coordinate (if it has one).
\begin{description}
\item[Dimensional Separability] For $i \in \{1, 2\}$, if $\pi_i(f)=\pi_i(g)$, $f \sim o$ and $g \sim o'$, then $\pi_i (o)=\pi_i (o')$.
\end{description}
This axiom neither implies nor is implied by Independence. But given
\begin{description}
\item[Certainty Equivalents] For all $f$, there is an $o$ such that $f \sim o$,
\end{description} 
it is immediately in conflict with Unidimensional Expectations, Pareto, Converse Pareto, and Negative Dominance. In fact nothing as strong as Unidimensional Expectations is required. The following much weaker assumption suffices:
\begin{description}
\item[Strict Betweenness] For unidimensional $f$, if there are $a, b$ in the support of $f$ such that $a \succ b$, then if $f \sim o$, there are $c, d$ in the  support of $f$ such that $c \succ o \succ d$.
\end{description}

\begin{proposition}\label{separability} Let $O= \mathbb R^2$. Pareto, Converse Pareto, Strict Betweenness, Certainty Equivalents, Dimensional Separability, and Negative Dominance are inconsistent.\end{proposition}

\begin{proof} By Strict Betweenness, Pareto, Converse Pareto, and Certainty Equivalents, for positive $a, b$, there are $x, y$ with $-a < x <a$ and $-b < y< b$ such that $(x, 0) \sim (-a,0)/(a,0)$ and $(0,y) \sim (0,b)/(0,-b)$. By Dimensional Separability, $(-a,b)/(-b,a) \sim (x,y)$. Given that $x$ and $y$ are strictly inside $[-a,a]$ and $[-b,b]$ there is an $\epsilon>0$ such that  $-a<x-\epsilon$ and $-b< y-\epsilon$. By Pareto, $(x-\epsilon, y-\epsilon) \prec (x,y)$ and hence $(x-\epsilon, y-\epsilon) \prec (-a,b)/(-b,a)$. By Converse Pareto,  $(x-\epsilon, y-\epsilon) \bowtie(-a,b)$ and $(x-\epsilon, y-\epsilon) \bowtie(-b,a)$, contradicting Negative Dominance. \end{proof}

This second impossibility result could be understood as evidence against Negative Dominance. Later, I will consider this response in more detail. But in what follows, I will mostly continue to focus on proposition \ref{main}, because in my view neither Dimensional Separability nor Certainty Equivalents are conceptually well motivated. Proposition \ref{expectationalism} motivates rejecting Expectationalism intuitively \emph{because} the latter implies Dimensional Separability. Proposition \ref{main}, by contrast, shows a conflict with principles which were not called into question by this opening result.

 \subsection{Relaxing Converse Pareto}\label{conversepareto}
   
I will understand proposition \ref{main} to indicate a conflict between: (i) incompleteness of preferences with a non-convex region of non-comparability; (ii) Negative Dominance; and (iii) Independence. To build the case for this interpretation, I show in this section that the result also holds for significant weakenings of Converse Pareto, and I say more to motivate these weakenings conceptually. In the subsequent two sections, I show that the result also holds for significant weakenings of Unidimensional Expectations.

Those who are satisfied with the idea that the conflict points to a tension between Independence and Negative Dominance in the present setting may wish to skip now to the next section, where I develop a theory which rejects Independence (section \ref{newtheory}).
  
Perhaps the most obvious way in which the impossibility result goes beyond the three minimal assumptions just stated is that Converse Pareto is significantly stronger than the assumption that incomplete preferences have a non-convex region of non-comparability. As noted earlier, Converse Pareto rules out non-trivial tradeoffs of any kind across dimensions. Those who endorse the rationality of incomplete preferences typically hold that, for generic $o,o'$ with $o \bowtie  o'$ there is often a \emph{sufficiently great} improvements in some dimension of $o'$ so that the resulting $o'^+ \succ o$. Converse Pareto rules out this possibility. If $o$ is worse along any dimension than $o'$, then any outcome which is the same in respect of that dimension as $o$ cannot be better than $o'$.  

Accordingly, it is natural to wonder whether it is this particular assumption about the structure of incomplete preferences, rather than the incompleteness itself, which drives the limitative results. But as we have seen, formally, the assumption is in fact inessential. Remark \ref{conversepareto1} already showed how proposition \ref{expectationalism} requires very little about the structure of incomparability. Proposition \ref{main} also allows substantial weakenings of Converse Pareto, allowing for nontrivial tradeoffs across dimensions. The only use of Converse Pareto is to ensure that $(-a/2, a)$ and $(a, -a/2)$ are both incomparable with $(0,0)$, and thus to obtain an inconsistency with Negative Dominance. But this assumption is motivated directly by the general idea that the decision-maker we are modeling does not have preferences defined on all tradeoffs between dimensions. Given this, it's natural to think that there could be some positive $a$ for which the decision-maker's preferences are not defined on the above pair. In fact, noted in remark \ref{uneven}, the proof does not depend on the values in the $x$-coordinate being the same as those in the $y$-coordinate. So we have already proven a much more general statement:

\begin{fact} \label{initialfact}Let $O= \mathbb R^2$. Suppose that for some $a,b>0$, $(a, -b/2)\bowtie  (0,0)$ and $(-a/2, b) \bowtie  (0,0)$. Then Pareto, Independence, Unidimensional Expectations and Negative Dominance are inconsistent. \end{fact}

I have said that the assumption that there are $a$ and $b$ such that $(a, -b/2)\bowtie  (0,0)$ and $(-a/2, b) \bowtie  (0,0)$ is ``plausible'', but one might wonder whether a systematic theory of partial comparability would preserve them. To see how we might reject Converse Pareto and provide a full theory of $\succeq$, which predicts incompleteness but allows non-trivial tradeoffs across dimensions, consider the lines $y=-2x$ and $y=-\frac{1}{2}x$, drawn in the following figure. 

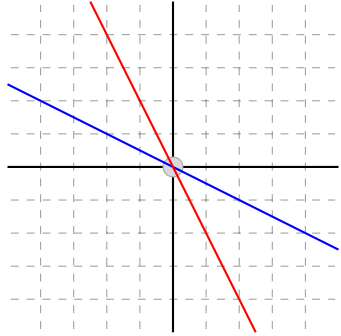
\begin{center}
\begin{tikzpicture}[ampersand replacement=\&,scale=.44]
  \draw[help lines, color=gray!80, dashed] (-4.9,-4.9) grid (4.9,4.9);
  \draw[-, thick] (-5,0)--(5,0);
  \draw[-, thick] (0,-5)--(0,5);
  \filldraw [thin, gray!70,fill=gray!30] (0,0) circle[radius=3mm];
  \draw[blue, thick] (-5,2.5)--(5,-2.5);
  \draw[red, thick] (-2.5,5)--(2.5,-5);
\end{tikzpicture}
\end{center}

We can use these two lines to define the set $\{ x | x \succeq (0,0)\}$, intuitively, as the points that are above and to the right of both lines. Similarly, we can define the set $\{ x | x \preceq (0,0) \}$, roughly, as the points that are below and to the left of both lines. The remaining points are understood to be incomparable with $(0,0)$. In the figure below, the green shaded points $\succeq (0,0)$, the purple shaded $\preceq (0,0)$, and the unshaded are incommensurable.

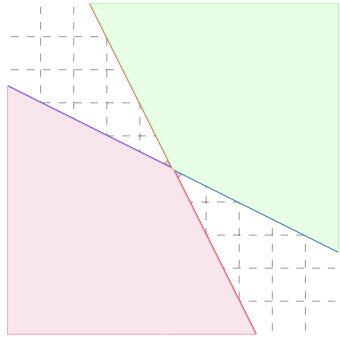
\begin{center}
\begin{tikzpicture}[ampersand replacement=\&,scale=.44]
  \draw[help lines, color=gray!80, dashed] (-4.9,-4.9) grid (4.9,4.9);
  \draw[-, thick] (-5,0)--(5,0);
  \draw[-, thick] (0,-5)--(0,5);
  \filldraw [thin, gray!70,fill=gray!30] (0,0) circle[radius=3mm];
  \draw[blue, thick] (-5,2.5)--(5,-2.5);
  \draw[red, thick] (-2.5,5)--(2.5,-5);
  \path[draw=green!30, fill=green!10] (-2.5,5)--(0,0)--(5,-2.5)--(5,5)--cycle;
  \path[draw=purple!30, fill=purple!10] (-5,2.5)--(0,0)--(2.5,-5)--(-5,-5)--cycle;
\end{tikzpicture}
\end{center}

In this concrete example, we chose two particular lines. But we can define a pre-order with any pair of negative slope lines, as follows:
\begin{description}
\item[Lines] For some positive $l, m$, with $l \neq m$, $(x,y) \succeq (x',y')$ if and only if $y \geq -l x + (l x' + y')$ \emph{and} $y \geq -m x + (m x' + y')$.
\end{description}

\begin{fact} Given Lines, $\succeq$ is transitive and reflexive. \end{fact}

The proof is by basic algebra.

The combination of Pareto and Converse Pareto can be seen as a limiting case of Lines, where $a$ is $0$, corresponding to horizontal lines, and $b$ is (inexactly, but intuitively) ``$\infty$'', corresponding to vertical lines. Lines generalizes this theory, allowing expansions of the domain of comparable points by a rotation of the vertical lines counterclockwise, and of the horizontal lines clockwise. 

Still, using essentially our argument above from Proposition \ref{main} we can still show directly that:

\begin{proposition}\label{lines} Let $O = \mathbb R^2$. Lines, Independence, Unidimensional Expectations, and Negative Dominance are inconsistent. \end{proposition}

\begin{proof} We choose $a, b$ so that both $(0,0) \bowtie  (-a/2,b)$ and $(0,0) \bowtie  (a, -b/2)$ (this is clearly possible, given any choice of slopes in Lines). We then run the argument of proposition \ref{main} starting with $(a,0)$, $(-a,0)$, $(0,b)$, $(0,-b)$.\end{proof}
 
Lines is one possible way that different dimensions may trade off nontrivially against one another, and it encompasses a wide array of such tradeoffs. Of course it may not be plausible for all relevant applications. But the proof here illustrates that the incompatibility does not depend essentially on Converse Pareto, and, moreover, would still hold for a wide variety of theories of partial comparability.

\subsection{Unidimensional Continuity and Certainty Equivalents}\label{continuity}

I have said that Unidimensional Expectations seems a possible description of some rational agent in our setting, but one might take our result to show that it is not, in the hopes of preserving both Independence and Negative Dominance. One might for example hold that either a rational decision-maker must have decreasing marginal utility in each of the relevant goods (even if there is no representation which allows us to speak of utility directly), or that the decision maker must be risk-averse or risk-prone.\footnote{Risk averse preferences famously violate Independence, so it is unclear whether this motivation really makes sense for a position that upholds Independence in response to our result. Still, it will help to guide us to a significant weakening of Unidimensional Expectations.}

Unfortunately, much weaker assumptions, which allow for such decision-makers, still suffice to generate the problem, if we assume Independence. The following two, slightly clunkier but much weaker assumptions turn out to be enough:
\begin{description}
\item[Unidimensional Continuity] If $a \succ b \succ c$ and $a,b,c$ are unidimensional with each other, there is an $f$ with support on $a,c$ such that $f \sim b$.
\item[Unidimensional Certainty Equivalents] If $f$ is unidimensional, then there is an $o^*$ which is unidimensional with $f$ such that $f \sim o^*$. Moreover, if there are $o, o'$ in the support of $f$ such that $o \succ o'$, then there are $o''$ and $o'''$ in the support of $f$ such that $o'' \succ o^* \succ o'''$.
\end{description}
These are both existence axioms, but of different kinds. Unidimensional Continuity requires the existence of certain lotteries, while Unidimensional Certainty Equivalents requires the existence of certain outcomes.\footnote{These axioms have the flavor of those which are shown by \citet{dubra2011continuity,karni2015continuity} to rule out incompleteness. But no analogue of the Archimedean axiom is used here, and the axioms are restricted to outcomes unidimensional with one another. Given that we have been happy to assume completeness for unidimensional lotteries, these axioms should not rule out incompleteness in general.}

The conjunction of these two axioms is a substantial weakening of Unidimensional Expectations. Given $a \succ b \succ c$, all unidimensional with one another, Unidimensional Expectations prescribes an exact probability for which $a \alpha c \sim b$. Unidimensional Continuity does not; it simply requires that there be some such probability. Similarly Unidimensional Expectations prescribes that any unidimensional lottery will be equivalent to its expectation (hence the name). But Unidimensional Certainty Equivalents does not: it simply says that the certainty equivalent must exist and lie between some of the outcomes in the support of the lottery. 

Together these axioms allow that the agent exhibits decreasing marginal utility in goods, and/or that they might be risk averse in unidimensional lotteries.

But once again in the presence of Independence, they do not allow us to escape the result:
\begin{proposition} \label{strongest} Let $O=\mathbb R^2$. Pareto, Converse Pareto, Unidimensional Continuity, Unidimensional Certainty Equivalents, Independence, and Negative Dominance are inconsistent. \end{proposition}

\begin{remark} The only use of Converse Pareto in the proof below will be in guaranteeing that two points $c \bowtie  (0,0)$ and $d \bowtie  (0,0)$, but we could equally plausibly assume this from the start. In particular, the inconsistency can also be generated given Lines.\end{remark}

\begin{proof} The proof is essentially as before. For any $a>0$, by Unidimensional Continuity, we have that there is an $\alpha$ such that$(a,0)\alpha(-a,0) \sim (0,0)$ and a $\beta$ such that $(0,a)\beta(0,-a) \sim (0,0)$. By Independence, $$f=[(a,0)\alpha(-a,0)]/[(0,a)\beta(0,-a)] \sim (0,0).$$ By Independence and Pareto, moving the leftmost point $(-a,0)$ upward, and the bottom point $(0,-a)$ rightward (while keeping probabilities the same) yields a better lottery. So we have $$f^{++}=[(a,0)\alpha(-a,a)]/[(0,a)\beta(a,-a)] \succ (0,0).$$ 
Expanding the notation, we have:
\begin{equation*}
\begin{aligned} & [(a,0)\alpha(-a,a)]/[(0,a)\beta(a,-a)]= \\ & \quad \qquad \qquad \frac{\alpha (a,0) + (1-\alpha)(-a,a)}{2} + \frac{\beta (0,a) + (1-\beta) (a,-a)}{2}.\end{aligned}\end{equation*}
Rearranging terms in the latter gives us:
$$\frac{\alpha (a,0) + (1-\beta) (a,-a)}{2} + \frac{\beta (0,a) +(1-\alpha) (-a,a)}{2}.$$
And this is now a mixture of two unidimensional lotteries. It is a mixture of (on the left ) $f_T=(a,0)\frac{\alpha}{\alpha + (1-\beta)}(a,-a)$ and (on the right) $f_B=(0,a) \frac{\beta}{\beta + (1-\alpha)}(-a,a)$ (``T'' for ``top'' and ``B'' for ``bottom''). In particular, we have $f^{++} = f_T \frac{\alpha + (1-\beta)}{2} f_B$. By Unidimensional Certainty Equivalents, we have that there is a $c$ such that $c \sim f_T$ and a $d$ such that $d \sim f_B$. So, by Independence, $f^{++} \sim c\frac{\alpha + (1-\beta)}{2}d$, and $c\frac{\alpha + (1-\beta)}{2}d \succ (0,0)$. By the second clause of Unidimensional Certainty Equivalents, $(a,0) \succ c \succ (a,-a)$ with $\pi_1(c)=a$, and $ (0,a) \succ d \succ (-a,a)$ with $\pi_2(d)=a$. So, by Converse Pareto, we have $c \bowtie  (0,0)$ and $d \bowtie  (0,0)$, contradicting Negative Dominance.
  \end{proof}

\subsection{Relaxing the requirement of certainty-equivalents}\label{kcomp}

Each of Unidimensional Expectations and Unidimensional Certainty Equivalents guarantees a certainty-equivalent for unidimensional lotteries, i.e. a certain outcome that is deemed indifferent to the lottery. But it is natural to explore whether the result holds also for a generalization of Unidimensional Expectations, which entails only that the lottery must be indifferent \emph{or incomparable} to its expected value (and similarly for Unidimensional Certainty Equivalents and Unidimensional Continuity). In other words, we might consider:

\begin{description}
\item[Weak Unidimensional Expectations] If $f$ is unidimensional, then for any $g$ if $exp(f) \succ g$, then $f \succ g$, and if $exp(f) \prec g$, $f \prec g$.
\end{description}

This axiom eliminates the requirement that $f \sim exp(f)$; the two may now be incomparable. It also eliminates two continuity assumptions of our original setup: first, that $\{ o | o \geq o'\}$ is closed in $\mathbb R^2$; second, that $\{ \alpha | f \alpha g \geq h\}$ is closed in $[0,1]$.\footnote{For more discussion of such continuity axioms, see \citet{schmeidler1971condition,galaabaatar2019completeness,dubra2011continuity,karni2015continuity,mccarthymikkola,badernew}.} But the main issue is that preserves the claim that for any we still have:

\begin{fact} Let $O= \mathbb R^2$. Pareto, Converse Pareto, Weak Unidimensional Expectations, Independence, and Negative Dominance are inconsistent. \end{fact}

The proof is again along the lines of proposition \ref{main}; it is also a corollary of the next fact.

In fact the argument can be generalized much further.

\begin{definition}
$\succeq$ satisfies \emph{k-incomparability} for $k \in \mathbb R$ if, for any unidimensional lottery $a/b$ (letting $n$, the distance between $a$ and $b$): \begin{itemize}
\item $a/b \succ c$ if $c$ lies $\frac{n}{k}$ or more ``below'' the expectation of $a/b$, on the line between $a$ and $b$;
\item $a/b \prec c$ if $c$ lies $\frac{n}{k}$ or more ``above'' the expectation of $a/b$, on the line between $a$ and $b$;
\item $a/b \bowtie  c$ if $c$ is in the open interval on the line between $a,b$ of length $n/2k$, centered at the expectation of $a/b$.
\end{itemize}
\end{definition}

This is a very weak constraint on a theory, since it only gives sufficient conditions of comparability for particular unidimensional points. Still, we can show that:

\begin{fact} Let $O = \mathbb R^2$. Pareto, Converse Pareto, Independence and Negative Dominance are incompatible if $\succeq$ satisfies k-incomparabilty.\end{fact}

\begin{remark} Nothing hinges on the particular assumptions about openness/closedness made in the definition of k-incomparability.\end{remark}

\begin{proof} We prove for the example of $k=4$ for ease of notation, though the proof generalizes to all $n$. By 4-incomparability, $(0,4)/(0,-4)$ is incomparable with every point in the open interval $((0,1),(0,-1))$, preferred to $(0,-1)$ and dispreferred to $(0,1)$. By 4-incomparability again $(4,0)/(-4,0)$ is incomparable with every point inside $((-1,0),(0,1))$, preferred to $(-1,0)$, and dispreferred to $(1,0)$. By Pareto, $(-1,-1)$ is strictly dispreferred to (-1,0) and (0,-1), so $(0,4)/(0,-4)$ and $(4,0)/(-4,0)$ are both strictly preferred to $(-1,-1)$. As above, two applications of Independence imply that $$(4,0)/(-4,0)/(0,4)/(0,-4)\succ (-1,-1).$$ 
So the improved $f^{++}= (4,0)/(-4,4)/(0,4)/(4,-4)$ is also strictly preferred to $(-1,-1)$. As before, note that $f^{++}$ is composed of two unidimensional lotteries: $(4,0)/(4,-4)$; and $(-4,4)/(0,4)$. By 4-incomparability, the first of these is incomparable with every point in the open interval $((4,-1.5),(4,-2.5))$, dispreferred to $(4,-1.5)$ and preferred to $(4,-1.5)$. By 4-incomparability again, the second is incomparable with every point in the open interval $((-1.5,4),(-2.5,4))$, preferred to $(-2.5,4)$ and dispreferred to $(-1.5,4)$. So again by two applications of Independence, we can derive that $f^{++} \prec (-1.5, 4) /(4,-1.5)$ (since each of these outcomes are preferred to each of the constituent unidimensional lotteries). So $$(-1.5,4)/(4,-1.5)\succ f^{++}=(4,0)/(-4,4)/(0,4)/(4,-4) \succ (-1,-1),$$
and thus $(-1.5,4)/(4,-1.5)\succ (-1,-1).$ By Converse Pareto, however $(-1,-1) \bowtie  (-1.5,4)$ and $(-1,-1) \bowtie  (4,-1.5)$, violating Negative Dominance.\end{proof}

The use of Converse Pareto is more demanding here than in earlier results, as I'll discuss below (Remark \ref{combining}).

\begin{remark}\label{combining}An appropriate combination of slopes in Lines and $k$ in a theory which satisfies k-incomparability is not obviously subject to the proof of this fact. For instance, it is not clear the argument can be given if $k=4$, $l=4$, and $m=\frac{1}{4}$. But, even if a consistency result could be given for these constraints, it would not yet be a general way around the results. The mere fact that the theory satisfies k-incomparability does not tell us how it handles a wide array of lotteries. The question of whether this combination can be extended to a more general theory of all lotteries is worthy of further investigation.\end{remark}

Say that the lottery $a/c$ \emph{spans} an outcome $b$ if and only if $a \succ b \succ c$. A plausible theory should predict that some outcomes are strictly preferred to some lotteries which span them (and strictly dispreferred to others). The previous fact makes it hard to see how a plausible systematic theory which makes this prediction could be consistent with Negative Dominance in the presence of Independence. This further builds the case that the conflict should be seen as between incompleteness Independence and Negative Dominance.   

\section{Negative Dominance without Independence}\label{newtheory}

I have suggested that the above limitative results point to a conflict between Independence on the one hand, and Negative Dominance on the other. As I discuss in subsection \ref{setsofutilities}, below, there are very well-known general theories of decision under uncertainty which uphold Independence, and thus must give up Negative Dominance. But there is not a similarly well-understood theory which upholds Negative Dominance and gives up Independence in the manner required by the impossibility results. Here, I develop a new theory which upholds Negative Dominance, while weakening Independence. The consistency of this theory is the second main result of the paper.

Our first goal is to show that, in the highly structured setting of $\mathbb R^2$, a strong theory including Negative Dominance is consistent.

The theory will include two new axioms. The first axiom requires a definition to state.

\begin{definition} A lottery is \emph{good} if and only if, for every $o, o'$ in the support of $f$, $o \parallel  o'$. \end{definition}

All unidimensional lotteries are good, but not all good lotteries are unidimensional. Visually, a unidimensional lottery's outcomes either all lie on a vertical line, or all lie on a horizontal line. There are many more good lotteries than this, however. For instance, if a lottery's outcomes all lie on a positive slope line---for instance, $g^*=(4,4)/(-2,-2)$---it is guaranteed to be good, although it is not unidimensional. Further lotteries are good too, without being unidimensional: for instance $(0,0)/(0,3)/(3,3)$. 

But of course not all lotteries are good; crucially the example from the introduction, $f^*=(4,-2)/(-2,4)$ is not good.

Our proposed axiom then is this:
\begin{description}
\item[Good Expectations] If $f$ is good then $f \sim exp(f)$.
\end{description}
This axiom only applies in a setting, like that of $\mathbb R^2$, where there is a well-defined notion of expectation for a lottery. We return to this point below.

Our next axiom is that lotteries which are stochastically dominant should be preferred. This assumption was implied by Independence but it is weaker than it. To introduce this formally we again need a definition.

Informally and intuitively, $f$ stochastically dominates $g$  if and only if, for every $o$, $f$ assigns at least as great a probability to outcomes weakly preferred to $o$ as $g$ does, and there is some outcome $o'$ for which $f$ assigns greater probability to outcomes weakly preferred to $o'$ than $g$ does. But \citet[\S 2.1]{russellfixing} shows that this definition leads to problems in a setting with incompleteness. To see this, consider first the lottery $f$ which assigns $\frac{2}{3}+\epsilon$ to $(2,2)$ and $\frac{1}{3}-\epsilon$ to $(0,0)$; and, second, the lottery $g=(2,0)/(0,2)/(2,2)$. Intuitively $f$ does not stochastically dominate $g$, because $f$ assigns $\frac{1}{3}-\epsilon$ to a worse outcome than any outcome of $g$. But the standard definition implies that $f$ \emph{does} stochastically dominate $g$, as the reader can easily verify.

To avoid this example, we must work instead with a more complex definition:

\begin{definition} \label{stochasticdominance}A \emph{generalized} lottery is a set $X \subset O \times [0,1]$. A generalized lottery $f^*$ is equivalent to a lottery $f$ iff for every outcome $o$ $\sum_{x \in {f^*}, \pi_1(x)=o} \pi_2(x)=f(o)$. 

A lottery $f$ \emph{stochastically dominates} a lottery $g$ in $\succeq$ iff there is a generalized lottery $f^*$ equivalent to $f$, a generalized lottery $g^*$ equivalent to $g$, and a bijection $f$ between them, such that for all $x \in f^*$ $\pi_2(x)=\pi_2(f(x))$ and $\pi_1(x) \succeq \pi_1 (f(x))$, and there is some $x \in f^*$ such that $\pi_1 (x) \succ \pi_1 (f(x))$. 
\end{definition}

Given this better definition, we can state our axiom:

\begin{description}
\item[Stochastic Dominance] If $f$ stochastically dominates $g$ in $\succeq$, then $f \succ g$.
\end{description}

Our central consistency result is then:

\begin{proposition}\label{consistency} Let $O=\mathbb R^2$, and assume Pareto and Converse Pareto. Good Expectations, Stochastic Dominance and Negative Dominance are consistent. \end{proposition}

The core of the proof of the proposition turns on the following Lemma:

\begin{lemma}\label{comparabilitylemma} Let $O=\mathbb R^2$, and assume Pareto and Converse Pareto. If $f$ is good, then there are $o^+$ and $o^-$ in the support of $f$ such that $o^+ \succeq exp(f) \succeq o^-$.\end{lemma}

\noindent Not every outcome in the support of a good lottery is comparable to its expectation. For instance, $(0,0)/(0,3)/(3,3)$ is good, but its expectation, $(1,1)$, is not comparable to $(0,3)$. But the proposition shows that, in general, there are always outcomes which lie on either side of the expectation (as $(0,0)$ and $(3,3)$ do here).

\begin{proof} If $f$ is good, its outcomes are totally ordered; let $o^+$ be its best outcome, and $o^-$ its worst. By Pareto and Converse Pareto, for all $o$ in the support of $f$ and $i \in \{1,2\}$ $\pi_i (o^+) \geq pi_i (o) \geq \pi_i (o^-)$. So for $i \in \{1,2\}$, $\pi_i (o^+) \geq pi_i (exp(f)) \geq \pi_i (o^-)$, and hence by Pareto, $o^+ \succeq exp(f) \succeq o^-$, as desired.\end{proof}

Using this Lemma, we now prove proposition \ref{consistency}:

\begin{proof} Let $\succeq$ be the minimal transitive reflexive relation on $\Delta(O)$ satisfying Pareto, Converse Pareto, Stochastic Dominance, and Good Expectations. 

We show that Good Expectations and Stochastic Dominance do not together force violations of Negative Dominance (on their own it is clear that they can't). Suppose $f$ is good and $f^+$ stochastically dominates $f$. If $f$ stochastically dominates $f^-$, then trivially there are $o, o'$ in the support of $f^+$ and $f^-$ respectively, so that $o \succeq o'$. If $f^-$ is good, and $exp(f) \succeq exp(f^-)$; we need to show again that in this case there are $o, o'$ in the support of $f^+$ and $f^-$ respectively, so that $o \succeq o'$. By lemma \ref{comparabilitylemma}, we have that there is $o^+$ in the support of $f$ such that $o^+\succeq exp(f)$, and $o^-$ in the support of $f^-$ so that $exp(f^-) \succeq o^-$. Since $f \succeq exp(f^-)$, we have that $o^+\succeq o^-$. Since $f^+$ stochastically dominates $f$, it assigns probability at least $f(o^+)$ to outcomes $o_1, \dots, o_n$ which are such that $o_i \succeq o^+ \succeq o^-$. Any of these outcomes suffices to ensure that $f^+ \succ f^-$ does not violate Negative Dominance. The same holds for $\prec$ as for $\succ$. \end{proof}

This consistency result shows the possibility of a strong theory, but we might hope for more. First, even in this particular setting, we might hope for a simple functional form for a utility function (or set of such functions) given preferences of this kind. I have been unable to find one. Second, we might hope to state the axioms on preferences in a way that does not depend on the structure of the space (i.e. on the existence of an expectation, as in Good Expectations). Third, we might hope to have a more unified conceptual framework that entails all three of Stochastic Dominance, Good Expectations, and Negative Dominance.

We can make some progress toward the second of these by considering the following axioms.

\begin{description}
\item[Continuity] For all $f, g, h$ the sets $\{ \alpha | f \preceq g \alpha h\}$ and $\{ \alpha | f \succeq g \alpha h\}$ are closed.
\item[Comparable Independence] For any $f, g, h$ if $o\parallel  o'$ for every $o, o'$ in the combined support of $f, g $, and $h$, then for any $\alpha$, $f \succeq g$ if and only if $f \alpha h \succeq g \alpha h$.
\end{description}

\begin{corollary} Pareto, Converse Pareto, Negative Dominance, Continuity, Comparable Independence, and Stochastic Dominance are consistent.\end{corollary} 

It is easy to verify that the order considered in the proof of proposition \ref{consistency} satisfies these assumptions.

This is some progress toward our goal of giving a strong theory which does not depend on the structure of our space of outcomes. But the theory is not as strong as one might have wished. We might have hoped that, if we add Unidimensional Expectations to the assumptions of the corollary, we would be able to \emph{derive} Good Expectations. Unfortunately I have not been able to resolve this question, and I conjecture that it does not hold.

\begin{conjecture} \label{badconjecture} Let $O=\mathbb R ^2$. Pareto, Converse Pareto, Negative Dominance, Continuity, Comparable Independence, Stochastic Dominance, and Unidimensional Expectations do not imply Good Expectations. \end{conjecture}

Thus, we are left with at least two important open questions: 
\begin{openq}
Are there natural axioms on preferences which do not make reference to the structure of the space of outcomes, but such that, together with Unidimensional Expectations imply Good Expectations? 
\end{openq}

\begin{openq} Is there a unified conceptual foundation for the package of Negative Dominance, Continuity, and Comparable Independence?
\end{openq}

Earlier, it was noted that the failures of Independence discussed here are quite different from the failures familiar from models of risk aversion following on the observations of Allais. The theory used in the proof of proposition \ref{consistency} allows striking failures of Independence. It says both that $[(-a,0)/(a,0)]/[(0,a)/(0,-a)]\not \sim (0,0)$ (even though each of its `halves' $\sim (0,0)$), and that $[(-a,a)/(0,a)]/[(a,0)/(a,-a)] \not \sim (-a/2,a)/(a,-a/2)$, again even though each of its `halves' $\sim$ the relevant outcomes). A typical explanation of failures of Independence inspired by Allais's example is that mixing via Independence can change the `security-value' of a lottery, and preferences are in part sensitive to this security value. I don't see how this explanation can be given in this case. So the failures of Independence predicted seem to be of a new kind, which requires a new justification.

There is also a third open question related to the failures of Independence. At least one of these two failures is required if the theory is to escape the limitative result of proposition \ref{main}. But a better theory would predict only one of the two failures. The first of the above failures of Independence can be ruled out by assuming
\begin{description}
\item[Indifferent Independence] If $f \sim g$ and $g \sim h$, then $f \alpha h \sim g \alpha h$.\footnote{This is inspired by the ``Betweenness'' of \citet{bottomleywilliamson}, a different condition than the axiom which standardly goes by that name. It is stronger than the usual ``linearity'' assumption, which says that if $f \sim g$, then $f \alpha g \sim f$.}
\end{description}
The third open question then is:
\begin{openq}
Is Indifferent Independence consistent with Good Expectations, Negative Dominance, and Stochastic Dominance?
\end{openq}

\subsection{Comparison to Sets of Utilities}\label{setsofutilities}

The focus of much of the literature on incomplete preferences has been the idea that such orders can be represented by sets of utility functions. Recently attention has focused on representation with sets of utilities, where uncertain options are evaluated by their expected utility (\citet{seidenfeld1995representation,shapleybaucells,dubra2004expected,nau2006shape,evren2008existence,evren2011multiutility,ok2012incomplete,galaabaatar2012expected,galaabaatar2013subjective,riella2015representation,gorno2017strict,hara2019coalitional,mccarthy2021expected,ok2023lipschitz,borie2023expected}). This framework assumes Independence. It thus offers us a ready-made option for those who wish to reject Negative Dominance in response to our earlier limitative results. 

Given the extensive work on this idea, it may seem the conservative choice in response to our results. It is interesting to ask how this theory compares in terms of strength to the one proven consistent above. In this section I take some steps toward answering that question.

Here I will understand the ``sets of utilities'' approach as follows. A function $u: \Delta(O) \to \mathbb R$ is \emph{mixture-preserving} if and only if for all $\alpha \in [0,1]$, $u (f \alpha g)=\alpha u (f) + (1-\alpha)u(g)$.
\begin{description}
\item[Sets of Utilities] There is a nonempty set $U$ of mixture-preserving functions from $\Delta(O)$ to $\mathbb R$, such that for all $f, g \in \Delta$, $f \succeq g \leftrightarrow u(f) \geq u(g)$ for all $u \in U$.
\end{description}

Independence follows immediately from the fact that all of the functions are mixture-preserving. We can thus see Sets of Utilities as a competitor to the theory of the previous section, which upholds Negative Dominance.

Sets of Utilities is consistent with Expectationalism and hence with Unidimensional Expectations:

\begin{fact} Let $O=\mathbb R^2$. Expectationalism, Pareto, and Converse Pareto, imply Sets of Utilities. \end{fact}
 
\begin{proof} When outcomes are ordered by Pareto and Converse Pareto, $f \sim exp(f)$ implies Sets of Utilities using all $u$ such that $u((x,y))=ax+by+k$. \end{proof}

In fact, Sets of Utilities and Unidimensional Expectations imply Expectationalism against this background:

\begin{proposition}\label{Rprop} Let $O=\mathbb R^2$. Pareto, Converse Pareto, Sets of Utilities and Unidimensional Expectations imply Expectationalism. \end{proposition}

\begin{proof} By Pareto, Unidimensional Expectations and Sets of Utilities, any $u\in U$ must have for any $c$, $d$, $u((x,d))=a_{d} x +k_{y=d}$ and $u((c, y))=b_c y + k_{x=c}$ where $a_d$ and $b_c$ are positive. We first show that for all $d$, $d'$, $a_d = a_{d'}$ (and similarly for $c$, $c'$, $b_c$ and $b_{c'}$). Suppose WLOG that $d'>d$. Then for all $x$ we must have $u((x,d')\geq u((x,d))$, that is $a_{d'} x +k_{y=d'}>a_d x + k_{y=d}$, that is, that $a_{d'} x> a_d x + (k_{y=d'} -k_{y=d'})$. To simplify notation we write $(k_{y=d'} -k_{y=d})$ as $\kappa_{d,d'}$. We show that this inequality requires that $a_{d'}=a_{d}$. (Recall that both must be positive.) If $a_{d'}>a_{d}$ we can choose a sufficiently small (negative) $x$ so that $a_{d'}x<a_d x+\kappa_{d,d'}$, contradicting Pareto and Sets of Utilities. And if $a_d>a_{d'}$, we can choose sufficiently large $x$ so that again $x<a_d x+\kappa_{d,d'}$, contradicting Pareto and Sets of Utilities. A similar argument holds for $b_c$. So we have that any utility function must have for any $c,d$ $u((x,d))=ax+k_{y=d}$ and $u((c,y))=by +k_{x=c}$.

By considering $u((0,0))$ we see in light of the foregoing that $k_{y=0}=k_{x=0}$. Call this $k$. We now show that $u((x,y))=ax+by+k$. Consider first $u((x^*,0))$. From the above, we have that $u((x^*,0))=ax^*+k$, but also that $u((x^*,0))=b\cdot0+k_{x=x^*}$, so that $k_{x=x^*}=ax^*+k$. For any $(x^*, y^*)$, we now have $u(x^*,y^*)=by^*+k_{x=x^*}=by^*+ax^*+k$, as desired.

By elementary algebra, for any such $u$, $u(f)=u(exp(f))$, and hence by Sets of Utilities, $f \sim exp(f)$.
\end{proof}

This result does depend crucially on Undimensional Expectations. In the absence of that assumption Sets of Utilities \emph{on its own} is consistent with Negative Dominance (and hence the negation of Expectationalism):

\begin{proposition}\label{utilityconsistency}Let $O=\mathbb R^2$. Pareto, Converse Pareto, Independence, Sets of Utilities and Negative Dominance are consistent. \end{proposition} 

\begin{proof} Consider the set of all utility functions such that $u((x,y))\geq u((x',y'))$ if $x\geq x'$ and $y\geq y'$. For any $x^*,y^*$, this set includes $u^+_{x^*,y^*}$ so that $u^+_{x^*,y^*}(x,y)=1$ if both $x\geq x^*$ and $y\geq y^*$ and $0$ otherwise. For any $x^*,y^*$, this set also includes $u^-_{x^*,y^*}$, so that $u^-_{x^*,y^*}(x,y)=0$ if $x \leq x^*$ and $y\leq y^*$, and $1$ otherwise. Consider an $f$ with support on a finite set of $o_1\dots o_n$, and such that all are incommensurable with $(x^*,y^*)$. Since the utility functions are required to be mixture-preserving, if a utility function assigns every outcome in the support of a lottery the same utility, it must assign the whole lottery that utility. We thus have $u^+_{x^*,y^*}(f)=0<1=u^+_{x^*,y^*}((x^*,y^*))$ and $u^-_{x^*,y^*}(f)=1>0=u^-_{x^*,y^*}((x^*,y^*))$, so by Sets of Utilities, $f \bowtie  o^*$ as desired.\end{proof}

As I said earlier, however, while Unidimensional Expectations seems to me ``optional'' (it is satisfied by some rational agents, but not others), it does seem to me a possible description of a rational agent. And any agent who satisfies it while also satisfying Independence must, as we have seen reject Negative Dominance.

Proposition \ref{Rprop} shows that Sets of Utilities is strong in a very important sense. Given Unidimensional Expectations alone, Sets of Utilities allows us to derive full Expectationalism. By contrast,  earlier I conjectured that the theory consisting of Negative Dominance, Continuity, and Comparable Independence does not allow us to derive Good Expectations from Unidimensional Expectations alone. This might make Sets of Utilities seem importantly stronger than the one which includes Negative Dominance.

In some sense this is correct. But in another sense it is not. Sets of Utilities is an assumption about representability of preferences, not a direct assumption about preferences themselves. Intrinsic assumptions about preferences which guarantee that they can be represented by such sets of utilities are not particularly well-motivated. They require, in addition to Independence and Continuity, assumptions about the structure of the underlying space on which they are defined. In fact our present setting violates known weak sufficient conditions on this structure to allow such a representation: $\mathbb R^2$ is not compact (\citet{dubra2004expected}, \citet{evren2008existence}); its dimension is uncountable (\citet{mccarthy2021expected}); and it does not satisfy \citet{mccarthy2021expected}'s ``countable domination'' property, as is easily verified.

I haven't been able to resolve whether Unidimensional Expectations, Independence and Continuity suffice for Sets of Utilities, but prior to the resolution of this question, it is not obvious that the preference-based axiomatization which includes Independence \emph{is} stronger than the one which includes Negative Dominance. In fact, I conjecture that the situation is exactly analogous to the situation described in conjecture \ref{badconjecture}:

\begin{conjecture} Let $O=\mathbb R^2$. Pareto, Converse Pareto, Independence, Continuity, and Unidimensional Expectations do not imply Expectationalism. \end{conjecture}

In sum, representability by Sets of Utilities gives us a strong constraint on preferences, stronger than what I have provided in the case of Negative Dominance. But the natural intrinsic properties of preferences are not strong in the same way. In our present state of knowledge the two theories are not dramatically different in strength.

\section{Related Work}\label{literature}

Many works have studied conditions for representing incomplete preferences with sets of utility functions, where preferences on uncertain options are given by sets of expected utilities (\citet{seidenfeld1995representation,shapleybaucells,dubra2004expected,nau2006shape,evren2008existence,evren2011multiutility,ok2012incomplete,galaabaatar2012expected,galaabaatar2013subjective,riella2015representation,gorno2017strict,hara2019coalitional,mccarthy2021expected,ok2023lipschitz,borie2023expected}). Independence (or something very close in spirit) is an assumption of all of these results. The present work supplements this tradition by focusing attention on an alternative, which rejects Independence with the aim of vindicating Negative Dominance.

\citet{manzini2008representation} (cf. \citet{manzini2004theory}) prove a representation theorem for incomplete preferences satisfying various ``sure thing'' principles, including the VST discussed immediately after Remark \ref{conversepareto1}. As noted there, their VST is strictly stronger than Negative Dominance, and is incompatible with Stochastic Dominance, given a rich domain of incomplete preferences like the one we have been assuming. Manzini and Mariotti's representation theorem, in terms of a utility function and a ``vagueness'' function, applies to the case where the set of non-comparable lotteries is concave in the set of all lotteries. The present work can be seen as complementing theirs, by exploring the case where incomplete preferences are not convex in this way, but where something quite similar to their VST (i.e. Negative Dominance) still holds. In our setting outcomes (not just lotteries) have their own affine structure, so there is more possibility for non-convexity of the region of non-comparability (indeed, this is what motivates the richness of our space). Moreover, we have here used weak preference as a primitive, whereas Manzini and Mariotti take strict preference as a primitive, and define indifference as a subset of the region which is non-comparable in terms of strict preference (for detailed investigation of the costs and benefits of this approach, see \citet{mandler2009indifference}). 

A point closely related to Manzini and Mariotti's observation about the limitations of VST was developed in detail by \citet{hare2010take} (for discussion see, among others \citet{hare2013limits, schoenfield2014decision,bales2014decision,bader2018stochastic,doody2019parity,doody2019opaque,doody2021hard,rabinowicz2021incommensurability,steele2021incommensurability}, \citet[\S 3.2]{russell2023fanaticism}). Hare works in an \citet{anscombe1963definition} or \citet{savage1972foundations}-inspired setting, with distinct states and outcomes, and a probability distribution on states; here (like Manzini and Mariotti) we have worked in a setting analogous to von Neumann-Morgenstern, with probability distributions defined directly on outcomes. Taking as given a finite set of states $S$, endowed with a probability measure $\mu$, Hare can be understood as taking the set of actions $A$ to be the set of functions from states to outcomes as usual. He then shows a conflict between Stochastic Dominance (modified to suit the setting with states in the obvious way; see \citet{bader2018stochastic}) and (what we here call) Statewise Negative Dominance:
\begin{description}
\item[Statewise Negative Dominance] If, for every state $s$, $a(s) \not \succ b(s)$, then $a \not \succ b$. 
\end{description}
This principle can be thought of as a version of VST, but developed in the Savage setting, with states in addition to outcomes.

The axiom leads to pathological behavior in the presence of incompleteness for exactly the same reasons as those discussed in the case of VST in section \ref{negdom}. In a setting with two states, $s_1$, $s_2$, where $\mu$ is the uniform distribution, we may have $a(s) \not \succ b(s)$ for all $s$, so Statewise Negative Dominance implies $a \not \succ b$, even though $a(s_1) \succ b(s_1)$ and $a(s_2) \succ b(s_2)$, so that Stochastic Dominance implies that $a \succ b$. (For instance, suppose $a(s_1)=(1,2), a(s_2)=(2,1)$, while $b(s_1)=(2,0), b(s_2)=(0,2)$.) Negative Dominance is substantially weaker than Statewise Negative Dominance; in the setting with states, the former is the weaker claim that, if every state $s, s'$ is such that $a(s) \not \succeq b(s')$ (note that $s'$ need not be identical to $s$), then $a \not \succ b$. Independence and Unidimensional Expectations are stronger than Stochastic Dominance, so the conflict is logically independent of Hare's. (For more discussion, see \citet[\S5]{ledermanmarbles}.)

\citet{pejsachowicz2017choice} study the connection between incomplete preferences over menus (in the tradition of \citet{kreps1979representation,dekel2001unforeseen,gul2001temptation}) and completions of those preferences. They prove a fascinating limitative result, showing that, given Independence and a continuity assumption, (roughly) if the completion of a possibly incomplete preference always changes incompleteness between two menus into a weak preference for the union of two menus over either of the menus, then either the original preference is complete, or the completion of the incomplete preference is monotonic in the sense that it always weakly prefers unions of pairs of menus to the menus themselves. One consequence of this is that the original incomplete preference cannot strictly prefer a menu $A$ to the menu $A \cup B$ for any $B$. The apparently weak property of the completion rules out natural properties of the incomplete preference, given Independence and a continuity assumption.

This surprising result has a similar flavor to the ones above, and they too suggest rejecting Independence as a live option. So our results may be seen as complementary to their general project. But the settings are quite different. \citet{pejsachowicz2017choice}'s key ``Cautious Deferral Completion'' property is conceptually motivated by, and formally essentially requires a setting in which preferences are defined on menus, not simple options. Similarly, the notion of monotonicity, which can be stated in the setting of menus, has no correlate in our setting (and therefore in our results). Finally, here I haven't drawn connections between completions of orders and our incomplete orders, but their statement requires one.

A different tradition studies social preferences derived from the the preferences of different individuals, where each individual may be understood as a ``dimension'' (working in the tradition of \citet{harsanyi1955cardinal}, \citet{mongin1994harsanyi}, see e.g. \citet{sen1970interpersonal,sen1973economic,sen1982choice}, and more recently \citet{danan2015harsanyi}, \citet{mccarthy2020utilitarianism}, \citet{danan2021partial}). Perhaps most closely related to our work \citet{danan2021partial} works in a setting where each individual's preferences are given by a \emph{set} of von Neumann-Morgenstern utility functions, and shows that (given certain background assumptions), the social welfare ordering cannot satisfy both completeness and Independence. Danan's individual dimensions are allowed to induce incomplete orders (whereas ours are incomplete), and an Independence of Irrelevant Alternatives axiom is imposed, which has no correlate in the setup here. Danan's representations violate Negative Dominance given particular input profiles, but this is not surprising, since he assumes Independence as a constraint on these representations (Theorem 2).

\section{Conclusion}\label{conclusion}

This paper has developed a series of results showing that (i) incomplete preferences on a sufficiently rich domain lead to a conflict between (ii) Negative Dominance and (iii) Independence. These three are inconsistent against a fairly plausible weak background. Some may see these results as adding to extant arguments against the rationality of incomplete preferences (\citet{gustafsson2022money}; \citet{dorr2021consequences}, \citet{dorr2021case}). Others may see them as providing an argument against Negative Dominance. Here, I have focused on the less well-understood possibility of rejecting Independence. I have provided a consistency result for a novel theory which rejects Independence, but also outlined several important open questions remaining for such a theory.

\begin{appendices}

\section{Qualitative Setting}\label{qualitative}

In this Appendix, I develop the limitative results of section \ref{r2} in a more abstract setting, where the dimensions of the outcomes are not assumed to be well-modeled by $\mathbb R$. The proofs are fairly straightforward generalizations of the proofs of preceding propositions, especially \ref{main}. But articulating the principles underlying what makes dimensions ``freely recombinable'' is I hope conceptually illuminating. Moreover, generalizing the results to this setting ends up being somewhat revealing mathematically: some of our strongest results in $\mathbb R^2$ are seen to have depended on strong assumptions about the richness of that space.

The results also further make the case that the core problem is the combination of (i) incomplete preferences on a sufficiently rich domain, (ii) Negative Dominance, (iii) Independence. The results show that the impossibility results depend less on the structure of the outcome space than might have seemed from section \ref{r2}.

I mentioned in the introduction that our basic formalism can be interpreted not as applying to preferences but as applying to an objective notion of ``betterness for a person''. In this Appendix I use that application in my informal discussion to illustrate it. Everything I say could also be reinterpreted (in my view equally plausibly) using the language of preferences.

\subsection{The Qualitative Setup}\label{basic2}

Prior to defining $O$, the set of outcomes, we take a step back. Conceptually, we take the pre-order $\succeq$ to be defined primitively over a non-empty set $W$, consisting of conceptually possible total descriptions of the universe. In line with our focus on multidimensional betterness (as opposed to preference), we understand there to be an exhaustive list (here assumed to be finite) of kinds of features which ``matter'' to whether the world is better for our decision-maker, at least in the sense that, if two worlds have exactly the same list of these features, then they are deemed equivalent $\sim$ by the order. 

Given this assumption, even though our underlying objects (or: outcomes) are ``really'' possible worlds, the aspects of outcomes which interest us are just the sequences of the relevant kinds of features, which correspond to equivalence classes of the worlds. We can order these sequences directly, since any two worlds which correspond to the same sequence are deemed equivalent in the order. So we will consider our set of outcomes $O$ to be the set of such sequences (or equivalently, equivalence classes of worlds, with respect to the sequences to which they correspond). Fixing $n$, the cardinality of dimensions, we will study spaces $O$ where outcomes can be written $(o_1, \dots, o_n)$. As before, we want to order not just $O$, but the set of lotteries $\Delta(O)\subset [0,1]^O$ where $f \in \Delta(O)$ if and only if it has finite support and $\sum_o f(o)=1$.

Given such a set of outcomes $O$, we can make sense of the set of features which are exhibited by some outcome on a given dimension, that is we can speak of $O_i=\{o_i \mid \exists o \in O$ such that $\pi_i(o)=o_i\}$ (where as before $\pi_i$ projects from a vector to its $i^{th}$ coordinate).

Given this basic setup, we begin with two background assumptions. The first concerns the set $O$:
\begin{description}
\item[Recombination] For $1 \leq i \leq n$, if $(o_1, \dots, o_i, \dots, o_n), (o'_1, \dots, o'_i, \dots, o'_n) \in O$,\\ then $(o_1, \dots, o'_i, \dots, o_n) \in O$. 
\end{description}
This assumption can also be formulated using the notion of $O_i$ introduced above, as the claim that $O=\times_{1 \leq i \leq n} O_i$. An analogue of this assumption was of course true in $\mathbb R^2$, but here we make it much more abstractly.

In the language of worlds, this assumption says that, if there is a world which exhibits some feature on a given dimension, then for any other world, there is a third world, which has the relevant feature of the first, and the rest of the profile of the second. This style of richness assumption articulates the idea that each dimension is at least conceptually separable from the others, so that we can make sense conceptually of an outcome obtained by recombining the relevant dimensions. There has been much discussion of such recombination principles in the literature on metaphysical modality (see, for an overview, \citet[\S 2.3]{menzel2023possible}). There it is broadly agreed that a principle of this kind could be plausible only if it concerns fundamental properties (and it might not be true even then). But here we are asking only for conceptual possibility (for an introduction to the difference, see \citet{kment2021varieties}). We will later see that even this assumption is stronger than what will be needed in our results.

Our second background assumption concerns the structure of $\succeq$. As we will understand this idea here, to say that aspects of the dimensions ``contribute'' to overall betterness for the decision-maker is to say that there is a preorder $\succeq_i$ on the set of features in each dimension $O_i$, so that, holding all other dimensions fixed, an improvement along one dimension will mean an improvement in the overall order as well. Formally:
\begin{description}
\item[Pareto] If for all $i$ $o_i \succeq_i o'_i$ then $(o_1, \dots, o_n) \succeq (o'_1, \dots, o'_n)$.
\end{description}
We choose our list of relevant features to be as long as we like. If Pareto appears to fail, then intuitively, we have not included all relevant dimensions; some relevant dimension must in fact not be improved. Note also that in speaking of improvements in relevant dimensions, we are not assuming that this corresponds to ``increases'' in the relevant dimension. 

Again in order to have crisp statements of the results I will also assume the converse of Pareto in the remainder of the Appendix:
\begin{description}
\item[Converse Pareto] If  $(o_1, \dots, o_n) \succeq (o'_1, \dots, o'_n)$ then for all $i$ $o_i \succeq_i o'_i$.
\end{description} 
As observed earlier, Converse Pareto is much more controversial than Pareto; I assume it here primarily to simplify the presentation. Just as in section \ref{r2}, it can be replaced with particular claims about incomparability without loss of plausibility.\footnote{In the philosophy literature, there is a conceptual challenge to Recombination and Pareto (see e.g. \citet{hedden2023dimensions}, with \citet{chang2002possibility,chang2017hard}). In certain cases, it may seem that increasing along one dimension disrupts the ``balance'' among all the dimensions. For instance, Ruth Chang suggests that: ``Adding a murder-mystery epilogue to Pride and Prejudice\textemdash result[s] in more literary merit in a respect, but a worse novel overall'' (\citet[402]{chang2016parity}). In response to this example, Hedden and Mu\~noz defend Pareto here by suggesting (as I did above) that if the example is genuine, we may simply consider a new dimension of ``balance''. But this move is not so straightforward in the presence of Recombination. If any dimension depends conceptually on other dimensions, then Recombination is threatened. In the case of balance, for instance, there cannot be ``good'' values of balance instantiated alongside very unbalanced other values. On the other hand, if we accept the example without ``balance'' as a dimension, then we would have to give up Pareto.}

\subsection{Unanimous Equivalence}\label{negdom2}

I'll work up to the strongest results in our qualitative setting through a series of weaker ones, which are analogous to those we developed for $\mathbb R^2$. The proofs of these results are not challenging; the interest is in formulating conceptually the ``right'' versions of axioms presented earlier.

In stating Pareto, we assumed preorders $\succeq_i$ on each $O_i$. In what follows, we assume that these $\succeq_i$ are defined on the whole of $\Delta (O_i)$. Then we can impose the following principle:
\begin{description}
\item[Unanimous Equivalence] If for all $i$ $a_i \alpha c_i \sim_i b_i$, then $(a_1 \dots a_n)\alpha (c_1 \dots c_n) \sim (b_1 \dots b_n)$.
\end{description}
A full Unanimity axiom would say that, if every $\succeq_i$ agrees on the ordering of the projection of two lotteries into $O_i$, then the overall ordering should agree as well. Our principle is restricted to cases where the individual orders treat the relevant lotteries as equivalent, and to comparing lotteries with binary support to a single outcome.

The axiom is in a sense an analogue of Expectationalism, though it is much weaker. The axiom says that if $b_i$ is a certainty equivalent of the lottery $a_i \alpha c_i$, then $(b_1 \dots b_n)$ is the certainty equivalent $(a_1 \dots a_n)\alpha (c_1 \dots c_n)$. We may think of $b_i$ as the coordinate-wise ``expectation'' of $a_i \alpha c_i$. Like Expectationalism, Unanimous Equivalence says that we can put together these coordinate-wise expectations to obtain an overall certainty equivalent.

But Unanimous Equivalence is much weaker than Expectationalism. It does not require that the agent's utilities (intuitively) be linear in each dimension, nor does it require that the agent is risk neutral. The axiom says nothing about which elements of $o_i$ act as certainty-equivalents for which pairwise lotteries. Still, it is enough to give us a result analogous to proposition \ref{expectationalism}.

\begin{proposition}\label{expectationalism2} Let $O=O_1 \times O_2$, so that Recombination holds. Assume that for $i \in \{1,2\}$ there are $a_i, b_i, c_i, d_i  \in O_i$ with $a_i \succ_i b_i \succ_i c_i \succ_i d_i$ and $a_i / d_i \sim b_i$. Then Pareto, Converse Pareto, and Unanimous Equivalence are inconsistent with Negative Dominance.
\end{proposition}

In this Appendix, we will often assume for simplicity, as here, that $n=2$, i.e. that there are just two dimensions relevant to $\succeq$. Nothing depends essentially on this assumption, as the reader may easily verify.

More importantly, in the proposition above, since we are working in a more abstract setting, we have had to write in by hand a further richness assumption beyond Recombination. The assumption is not very strong at all: it just says that there four strictly ranked properties in two of the relevant dimensions, and a uniform lottery over the best and worst which is equivalent to the better of the intermediate options.

\begin{proof} By Unanimous Equivalence, $(a_1, d_2) / (d_1,a_2) \sim (b_1, b_2)$ and thus by Pareto $(a_1, d_2) / (d_1,a_2) \succ (c_1, c_2)$. By Converse Pareto, we have that $ (c_1, c_2) \bowtie  (a_1, d_2)$ and $(c_1, c_2) \bowtie  (a_2, d_1)$, contradicting Negative Dominance.
\end{proof}

\begin{remark} The use of the uniform lottery $a_i/d_i$ was just to simplify the presentation. We could have instead assumed only that for some $\alpha$ $a_1 \alpha d_1 \sim_1 b_1$ and $a_2 (1-\alpha) d_2 \sim_2 b_2$. \end{remark}

\begin{remark}\label{conversepareto2} As before, the use of Converse Pareto is inessential. We could have instead assumed directly that $ (c_1, c_2) \bowtie  (a_1, d_2)$ and $ (c_1, c_2) \bowtie  (a_2, d_1)$.\end{remark}

\subsection{Dimensional Separability}\label{unidimensional2}

So far, the analogues of our ``expectational'' assumptions in the qualitative setting do not rely on the notion of an expectation, but more abstractly on certainty-equivalents, and on the relationship between the several $\succeq_i$ and $\succeq$. Something similar will turn out to be true for the analogue of our main result, although in this setting there is not such a strong distinction to be drawn between a result based on ``Unidimensional Expectations'', and one based on a separability assumption (as in proposition \ref{separability}).

In particular, the idea will be to require that, if $\succeq_i$ orders lotteries in $\Delta(O_i)$ in a particular way, then any unidimensional lotteries which have this projection into $i$, will be ordered in the same way. To state this idea formally, we start with some basic notation. For any $1 \leq i \leq n$ we write $O_{-i} = O_1 \times \dots \times O_{i-1} \times O_{i+1} \times \dots O_n$. Given $o_{-i} \in O_{-i}$ and $o_i \in O_i$ we write $(o_i,o_{-i})$ for the element $o \in O$ such that $\pi_i (o)=o_i$, and for all $1 \leq j \leq n$, with $j\neq i$ $\pi_j (o)=\pi_j(o_{-1})$. Similarly, given a lottery just over elements of $O_i$, $f_i \in \Delta (O_i)$, we write $(f_i, o_{-i})$ for the unique $f \in \Delta(O)$ such that for all $o_i$, $f(o_i, o_{-i})=f_i(o_i)$. 

Using this notation, we state our substantive assumption:
\begin{description}
\item[Unidimensional Dimensional Separability] For all $f_i, g_i \in \Delta (O_i)$, $f_i \succeq_i g_i$ if and only if for all $o_{-i} \in O_{-i}$, $(f_i, o_{-i}) \succeq (g_i, o_{-i})$.
\end{description}

This axiom is in an important sense weaker than the earlier Dimensional Separability, because it applies only to unidimensional lotteries. But it also differs in some subtler ways: by applying more generally than to lotteries with certainty equivalents, and also making reference to an ordering on projections of lotteries, which we did not have in the earlier setting.

This axiom does not follow from Unanimous Equivalence. But it is intuitively weaker. The only reason it does not follow is that Unanimous Equivalence applied only to cases where the individual orders were indifferent between a pair of lotteries, whereas the present one allows the case where a given dimension $i$ may have a strict inequality. Otherwise, the new axiom is substantially weaker, since it only applies in cases where the only non-trivial variation is along a single dimension. It rules out the possibility that, as values in one of the other dimensions changes, the valuing of unidimensional lotteries might be affected. But otherwise it is extremely permissive.

Even so, it turns out to be inconsistent with Negative Dominance given modest background assumptions. To see this, we start by stating a corollary of the main result, since it is easier to parse:

\begin{corollary} Let $O= O_1 \times O_2$ so that Recombination holds. Suppose that for $i \in \{1,2\}$ there are $a_i, b_i, c_i, d_i  \in O_i$, such that $a_i \succ_i b_i \succ_i c_i \succ_i d_i$, $a_i / d_i \sim_i b_i$, and $c_i \succeq_i b_i/ d_i$. Then Pareto, Converse Pareto, Independence, and Unidimensional Dimensional Separability are inconsistent with Negative Dominance.\end{corollary}

The assumption that there are such dimensions with such $a_i,b_i,c_i,d_i$ is very weak. But we can weaken it even further. The corollary deals with uniform lotteries over two outcomes, but the proposition generalizes this to other lotteries with probabilities $\alpha$ and $\beta$ (above $\alpha = \beta = \frac{1}{2}$):

\begin{proposition}\label{main2}Let $O= O_1 \times O_2$ so that Recombination holds. Suppose there are $\alpha, \beta \in (0,1)$ and for all $i \in \{1,2\}$ there are $a_i, b_i, c_i, d_i  \in O_i$,  such that $a_i \succ b_i \succ c_i \succ d_i$, $a_j \alpha d_j \sim_j b_j$, $c_j \succeq_j b_j \frac{\beta}{\beta + (1-\alpha)} d_j$, $a_k \beta d_k \sim_k b_k$ and $c_2 \succeq b_k \frac{\alpha}{\alpha + (1- \beta)} d_k$. Then Pareto, Converse Pareto, Independence, and Unidimensional Dimensional Separability are inconsistent with Negative Dominance.\end{proposition}

\begin{proof} The proof is exactly parallel to that of proposition \ref{main}, using $(b_1,b_2)$ in place of $(0,0)$, $(a_1,b_2)$ in place of $(a,0)$, $(d_1,b_2)$ in place of $(-a,0)$, $(b_1,a_2)$ in place of $(0,a)$ and $(b_1,d_2)$ in place of $(0,-a)$. The proof proceeds as in the proof of \ref{main}, appealing in the final stage to the fact that $(b_1, a_2)/(d_1,a_2) \preceq (c_1, a_2)$ and $(a_1,b_2)/(a_1,d_2)\preceq (a_1,c_2)$.\end{proof}

\begin{remark}\label{conversepareto4} As before, the use of Converse Pareto is inessential. We could have assumed directly that $(c_1, a_2) \bowtie  (b_1, b_2)$, and $(a_1, c_2) \bowtie  (b_1, b_2)$. It is unclear what systematic view of $\succeq$ would rule out this possibility. \end{remark}

\subsection{Unidimensional Continuity and Certainty Equivalents}\label{continuity2}

Earlier, I built the case for focusing on Independence and Negative Dominance by considering weakenings of Unidimensional Expectations, and showing that they too give rise to the result. In the present qualitative setting, I have used Unidimensional Dimensional Separability, which is compatible with both diminishing marginal utility and some form of risk aversion. At the same time, as mentioned earlier in connection to Dimensional Separability, it is not obvious how plausible the separability assumption is.  

In fact, as I will now discuss, it is not just that our earlier motivation for studying Unidimensional Continuity and Unidimensional Certainty Equivalents does not apply in the qualitative setting. It turns out that, in the present setting, without the background structure of $\mathbb R^2$, these assumptions are not conceptually weaker than Unidimensional Dimensional Separability.

To see this, we must state analogues of Unidimensional Continuity and Unidimensional Certainty Equivalents. In the present qualitative setting, an outcome $o$ is \emph{unidimensional with} an outcome $o'$ if and only if there is at most one $i$ with $1 \leq i \leq n$ such that $\pi_j(o) \neq \pi_j (o')$. The notions of lotteries being unidimensional with one another, and of a unidimensional lottery, are as before. We have:
\begin{description}
\item[Qualitative Unidimensional Continuity] If $a \succ b \succ c$ and $a,b,c$ are unidimensional with one another, there is an $f$ with support on $a,c$ such that $f \sim b$.
\item[Qualitative Unidimensional Certainty Equivalents] If $f$ is unidimensional, then there is an $o^*$ which is unidimensional with $f$ and such that $f \sim o^*$. If there are $o, o'$ in the support of $f$ such that $o \succ o'$, then there are $o''$ and $o'''$ in the support of $f$ such that $o'' \succ o^* \succ o'''$.
\end{description}

\noindent Before discussing these assumptions conceptually, I present a version of proposition \ref{strongest} using them:

\begin{proposition}\label{strongest2} Let $O= O_1 \times O_2$ so that Recombination holds. Suppose that for all $i \in \{ 1, 2\}$, there are $a_i, b_i, c_i  \in O_i$ with $a_i \succ_i b_i \succ_i c_i$. Then Pareto, Converse Pareto, Independence, Qualitative Unidimensional Continuity and Qualitative Unidimensional Certainty Equivalents are inconsistent.\end{proposition}

\begin{proof} The proof is exactly analogous to that of proposition \ref{strongest}, using $(b_1,b_2)$ in place of $(0,0)$, $(a_1,b_2)$ in place of $(a,0)$, $(c_1,b_2)$ in place of $(-a,0)$, $(b_1,a_2)$ in place of $(0,a)$ and $(b_1,c_2)$ in place of $(0,-a)$. The proof proceeds exactly as above.\end{proof}

This proposition may appear considerably weaker than proposition \ref{main2}. We only assume three strictly ordered elements of $O_1$ and $O_2$, and Qualitative Unidimensional Continuity and Qualitative Unidimensional Certainty Equivalents allow us to do away with specific assumptions about the equivalence of lotteries to outcomes. But this appearance is to some extent misleading. The reason is that Qualitative Unidimensional Certainty Equivalents is a very powerful axiom. Together with Independence, it implies that between any two strictly ordered outcomes which are unidimensional with one another, there is a continuum of other outcomes. If $a \succ b$ and $a$ is unidimensional with $b$ then Independence implies that for any $\alpha, \beta \in (0,1)$, with $\alpha > \beta$, that $a \alpha b \succ a \beta b$. Qualitative Unidimensional Certainty Equivalents implies that there are $c, d$ unidimensional with $a$ and $b$ such that $a\alpha b \sim c$ and $a\beta b \sim d$. So, if there is a pair of outcomes which are unidimensional with one another, with a strict preference between them on one dimension, then between them, there is a continuum of other outcomes. 

$\mathbb R^2$ already has this denseness built into it, so the fact that Unidimensional Certainty Equivalents had this consequence was not relevant in the earlier setting. But in the present setting, Qualitative Unidimensional Certainty Equivalents enforces substantial new structure on the background space.

Still, Qualitative Unidimensional Certainty Equivalents is only used in the proof to guarantee that there are incomparable outcomes which are certainty-equivalent to the relevant lotteries between $(c_1,a_2)$ $(b_1,a_2)$ and $(a_1,b_2)$ $(a_1,c_2)$. We might assume this directly without loss of plausibility, and without in general guaranteeing the denseness just mentioned.  And if we do the proposition is a significant improvement on the previous one. In addition to its using weaker richness assumptions (as I just noted), the assumptions used here are weaker than Unidimensional Dimensional Separability in other ways. Most importantly, they do not require that we can rank unidimensional lotteries according to how they behave in one dimension, without regard for the values in others.

\subsection{Conclusion}

The propositions in this Appendix are generalizations of those in section \ref{r2}, although not strictly in a logical sense. Unidimensional Dimensional Separability is strictly weaker than Unidimensional Expectations in $\mathbb R^n$, and Recombination is automatically satisfied there. But the richness assumptions made directly in proposition \ref{main2} have no correlate in the parallel proposition \ref{main}. Instead, features of $\mathbb R^n$, together with the strength of Unidimensional Expectations entail that they will be satisfied. Still the propositions do intuitively apply considerably more generally. For it is hard to imagine a systematic, plausible way in which Dimensional Separability would hold in $\mathbb R^n$, without these richness conditions being satisfied. And of course the qualitative setting has much less structure than $\mathbb R^n$.

One obvious way in which this is so, is that, whereas each dimension in $\mathbb R^n$ is cardinally (and indeed ratio) measurable, we have only assumed ordinal structure directly on the dimensions in this more abstract setting. It might have been natural to expect that, in moving to the qualitative setting, it would be necessary to impose enough structure on each dimension that they would be (at least) cardinally measurable. But the results do not depend on such structure; they only require that particular lotteries have certainty-equivalents. Guaranteeing such certainty-equivalents requires substantive assumptions, but quite different ones from those which ensure (for instance) cardinal structure.

\end{appendices}

\begin{small}
\begin{singlespacing}

\bibliographystyle{plainnat}
\bibliography{references}
\end{singlespacing}

\end{small}

\end{document}